\documentclass[11pt,a4paper]{article}

\usepackage[utf8]{inputenc}
\usepackage[T1]{fontenc}
\usepackage{lmodern}

\usepackage{amsmath,amssymb,amsthm}
\usepackage{bm}
\usepackage{mathtools}

\usepackage{booktabs}
\usepackage{array}
\usepackage{tabularx}
\usepackage{float}

\usepackage[margin=2.5cm]{geometry}
\usepackage{setspace}
\usepackage{parskip}

\usepackage[authoryear,round]{natbib}

\usepackage{tikz}
\usepackage{pgfplots}
\pgfplotsset{compat=1.18}
\usetikzlibrary{arrows.meta, decorations.markings, patterns, calc}

\usepackage[colorlinks=true,linkcolor=blue,citecolor=blue,urlcolor=blue]{hyperref}
\usepackage{cleveref}

\newtheorem{proposition}{Proposition}
\newtheorem*{proposition*}{Proposition}

\theoremstyle{definition}
\newtheorem{definition}{Definition}
\theoremstyle{remark}
\newtheorem*{remark}{Remark}

\newcommand{\w}{\mathbf{w}}

\newcommand{\one}{\mathbf{1}}
\newcommand{\V}{\mathbf{V}}

\newcommand{\Sig}{\boldsymbol{\Sigma}}
\newcommand{\muvec}{\boldsymbol{\mu}}

\newcommand{\tw}{\tau_w}
\newcommand{\rf}{r_f}

\title{Extensions to the Wealth Tax Neutrality Framework}
\author{Anders G Fr{\o}seth\thanks{Independent Researcher.
  E-mail: \href{mailto:indrefjorden@pm.me}{indrefjorden@pm.me}.}}
\date{\today}

\begin{document}
\maketitle

\begin{abstract}
\noindent
\citet{Froeseth2026} shows that a proportional wealth tax on market
values is neutral with respect to portfolio choice, Sharpe ratios, and
equilibrium prices under CRRA preferences and geometric Brownian motion.
This paper investigates the robustness of that result along two
dimensions.  First, we extend the neutrality frontier: portfolio
neutrality---including all intertemporal hedging demands---is preserved
under stochastic volatility (Heston and general Markov diffusions) and
Epstein--Zin recursive utility, but breaks under non-homothetic
preferences such as HARA.  Second, we identify four channels through
which implemented wealth taxes depart from neutrality even under CRRA:
non-uniform assessment across asset classes, general equilibrium price
effects in inelastic markets, progressive threshold structures, and
endogenous labour supply.  Each channel is
formalised and, where possible, calibrated to the Norwegian wealth tax
system.  The progressive threshold introduces a tax shield that
\emph{increases} risk-taking near the exemption boundary---an effect
opposite in sign to the HARA distortion---and, at the extreme,
generates a participation margin at which investors exit the tax
jurisdiction entirely.  We formalise this tax-induced migration as the
extreme response at the progressive threshold and examine the Norwegian
post-2022 experience as a case study.  The full framework is applied to
evaluate the Saez--Zucman proposal for a global minimum wealth tax on
billionaires and the related French proposal for a national minimum tax
above \texteuro100 million.
\end{abstract}

\medskip
\noindent\textbf{JEL Classification:} G11, G12, H21, H24, H26, H73.

\smallskip
\noindent\textbf{Keywords:} Wealth tax, portfolio choice, tax neutrality,
stochastic volatility, HARA preferences, progressive taxation, inelastic
markets, tax-induced migration, Saez--Zucman proposal, Norway.

\thispagestyle{empty}
\setcounter{page}{2}

\section{Introduction}\label{sec:intro}

This paper builds on the wealth tax neutrality framework developed in
\citet{Froeseth2026}, which established that a proportional wealth tax on
market values is neutral with respect to portfolio choice, Sharpe ratios,
and equilibrium prices.  That analysis proceeds under geometric Brownian
motion and then extends to the location-scale family of return
distributions.  Here we investigate whether the neutrality results survive
under more general conditions, and systematically identify the channels
through which real-world wealth taxes depart from neutrality.

The paper has three parts.  The first
(Sections~\ref{sec:stochvol}--\ref{sec:beyond_crra}) extends the
neutrality frontier: we show that portfolio neutrality survives under
stochastic volatility, general Markov diffusion dynamics, and
Epstein--Zin recursive utility, while it breaks under non-homothetic
preferences such as HARA utility.  The second
(Sections~\ref{sec:nonuniform}--\ref{sec:labour}) analyses the
channels through which implemented wealth taxes depart from neutrality,
even when investors have CRRA preferences: non-uniform assessment across
asset classes, general equilibrium price effects in inelastic markets,
progressive threshold structures, and endogenous labour supply.
The third part applies and extends the framework.
Section~\ref{sec:zucman} synthesises all channels to evaluate
two recent proposals: the Saez--Zucman global minimum wealth tax on
billionaires and the French national minimum
tax above \texteuro100 million.
\Cref{sec:migration} formalises tax-induced migration as the extreme
participation response at the progressive threshold, using the Norwegian
post-2022 experience as a case study.

\Cref{sec:stochvol} considers stochastic volatility models, where the
variance of returns is itself a random process.  Using the Heston model as
the leading example, we show that portfolio neutrality---including the
intertemporal hedging demand identified by \citet{Merton1973}---is
preserved under CRRA preferences.  The result generalises to any Markov
diffusion model of asset returns.
\Cref{sec:beyond_crra} identifies CRRA as the operative condition:
non-homothetic preferences (HARA, wealth-in-utility) break neutrality,
while Epstein--Zin preferences preserve it.
\Cref{sec:nonuniform} derives the portfolio distortion from
asset-class-specific assessment discounts, calibrated to the Norwegian
system.
\Cref{sec:inelastic} analyses general equilibrium price effects through
the inelastic markets hypothesis.
\Cref{sec:progressive} formalises the distortion from progressive
taxation (thresholds and brackets), and
\Cref{sec:labour} extends the framework to endogenous labour supply and
entrepreneurial effort.
\Cref{sec:zucman} synthesises all channels in a comparative evaluation
of the Saez--Zucman global proposal and the French national variant.
\Cref{sec:migration} formalises tax-induced migration as the extreme
participation response at the progressive threshold and examines the
Norwegian post-2022 experience as a case study.

\Cref{tab:results_summary} provides an overview of the main results.

\begin{table}[H]
\centering\small
\begin{tabular}{llcc}
\toprule
Section & Result & Proposition & Sign \\
\midrule
\multicolumn{4}{l}{\emph{Part I: Neutrality frontier}} \\
\ref{sec:stochvol}
  & Stochastic volatility (Heston, Markov)
  & \ref{prop:sv_neutrality}, \ref{prop:markov_neutrality}
  & Neutral \\
\ref{sec:beyond_crra}
  & Epstein--Zin recursive utility
  & \ref{prop:ez}
  & Neutral \\
\ref{sec:beyond_crra}
  & HARA (non-homothetic) preferences
  & \ref{prop:hara_distortion}
  & $\Delta w^* < 0$ \\[4pt]
\multicolumn{4}{l}{\emph{Part II: Non-neutrality channels}} \\
\ref{sec:nonuniform}
  & Non-uniform assessment
  & \ref{prop:nonuniform}
  & Asset-dependent \\
\ref{sec:inelastic}
  & Inelastic markets (GE price impact)
  & ---
  & $\Delta P/P < 0$ \\
\ref{sec:progressive}
  & Progressive threshold (tax shield)
  & \ref{prop:progressive}
  & $\Delta w^* > 0$ \\
\ref{sec:labour}
  & Endogenous labour supply
  & \ref{prop:labour_separability}
  & Separable$^{\dagger}$ \\[4pt]
\multicolumn{4}{l}{\emph{Part III: Applications and extensions}} \\
\ref{sec:zucman}
  & Zucman/French proposals
  & ---
  & Synthesis \\
\ref{sec:migration}
  & Tax-induced migration
  & ---
  & Exit at $W_i^*$ \\
\bottomrule
\end{tabular}

\smallskip
{\footnotesize $^{\dagger}$Separable under proportional taxation;
distorted at progressive thresholds.}
\caption{Summary of main results.  ``Sign'' indicates the direction
of the portfolio or price distortion relative to the no-tax benchmark
under CRRA preferences.}
\label{tab:results_summary}
\end{table}

\section{Stochastic Volatility}\label{sec:stochvol}

The neutrality results in \citet{Froeseth2026} are derived first under
geometric Brownian motion (constant $\mu$, $\sigma$) and then generalised
to the location-scale family of return distributions.  A natural question
is whether the results extend to stochastic volatility models, where the
variance of returns is itself a random process.  This is empirically
important: asset returns exhibit time-varying volatility, fat tails, and
leverage effects that are absent under GBM
\citep{Heston1993,Yakovenko2002}.

We show that, under CRRA preferences, portfolio neutrality extends to the
Heston stochastic volatility model---and, more generally, to any Markov
diffusion model of asset returns---\emph{without} requiring the
location-scale property.  The mechanism is different from the
location-scale argument: it relies on the homogeneity of the CRRA value
function in wealth, which makes optimal portfolio weights (including the
intertemporal hedging demand identified by \citealt{Merton1973})
independent of the wealth level.

\subsection{The Heston model}\label{sec:heston_setup}

The Heston model specifies the price of a single risky asset and its
instantaneous variance as a pair of coupled diffusions:
\begin{align}
  \frac{dS}{S} &= \mu \, dt + \sqrt{v_t} \, dW_t^{(1)},
  \label{eq:heston_price}\\[4pt]
  dv_t &= \lambda(\theta - v_t) \, dt + \kappa \sqrt{v_t} \, dW_t^{(2)},
  \label{eq:heston_var}
\end{align}
where $v_t = \sigma_t^2$ is the instantaneous variance, $\theta > 0$ is
the long-run mean variance, $\lambda > 0$ is the rate of mean reversion,
$\kappa > 0$ is the volatility of variance, and
$\mathrm{corr}(dW^{(1)}, dW^{(2)}) = \rho$.  When $\rho < 0$, negative
returns coincide with rising volatility---the leverage effect.  The Feller
condition $2\lambda\theta \geq \kappa^2$ ensures the variance remains
strictly positive.

\begin{remark}[Heston returns and the location-scale family]
\label{rem:heston_not_ls}
\citet{Yakovenko2002} derive the probability density of log-returns under
the Heston model in closed form and show that it involves a modified
Bessel function $K_1$ of a scaled argument.  The distribution exhibits
exponential (rather than Gaussian) tails for large returns, and asymmetry
when $\rho \neq 0$.  These features place Heston returns \emph{outside}
the location-scale family.  The generalised Propositions~2$'$ and~3$'$ in
\citet{Froeseth2026}, which require the location-scale property, therefore
do not apply directly.
\end{remark}

\subsection{Wealth dynamics with a wealth tax}\label{sec:heston_wealth}

An investor allocates a fraction $w$ of wealth to the risky asset and the
remainder $1 - w$ to a risk-free asset earning a continuous rate $\rf$.
Following the continuous-time formulation in \citet{Froeseth2026}, the
proportional wealth tax $\tw$ enters as an additional drain on the drift
of wealth.  The investor's wealth evolves as
\begin{equation}\label{eq:heston_dW}
  dW = \bigl\{W[\rf + w(\mu - \rf) - \tw] - C\bigr\} \, dt
       + w W \sqrt{v_t} \, dW_t^{(1)},
\end{equation}
where $C \geq 0$ is the consumption flow rate.  As in the GBM case, the
tax reduces the drift of wealth by $\tw W$ but leaves the diffusion
coefficient---and hence the per-share return $dS/S$---unchanged.  The
variance state variable $v_t$ evolves according to~\eqref{eq:heston_var},
independently of the tax.

\subsection{The Merton problem under Heston}\label{sec:heston_merton}

Consider an investor with CRRA utility over consumption,
\begin{equation}\label{eq:crra}
  U(C) = \frac{C^{1-\gamma}}{1-\gamma}, \qquad \gamma > 0, \;
  \gamma \neq 1,
\end{equation}
who maximises expected discounted lifetime utility
$E\!\left[\int_0^T e^{-\delta t} U(C_t) \, dt\right]$
subject to the wealth dynamics~\eqref{eq:heston_dW} and the variance
dynamics~\eqref{eq:heston_var}.

The value function is
\begin{equation}\label{eq:value_fn}
  J(W, v, t) = \max_{C,w} \;
  E_t\!\left[\int_t^T e^{-\delta(s-t)} U(C_s) \, ds\right].
\end{equation}
By Bellman's principle, $J$ satisfies the Hamilton--Jacobi--Bellman (HJB)
equation
\begin{equation}\label{eq:hjb}
\begin{split}
  0 = \max_{C, w} \biggl\{
    &U(C) - \delta J + J_t
    + J_W \bigl[W(\rf + w(\mu - \rf) - \tw) - C\bigr] \\
    &+ J_v \, \lambda(\theta - v)
    + \tfrac{1}{2} J_{WW} \, w^2 v W^2
    + \tfrac{1}{2} J_{vv} \, \kappa^2 v
    + J_{Wv} \, w W \kappa v \rho
  \biggr\}.
\end{split}
\end{equation}

\subsection{CRRA value function and optimal portfolio}
\label{sec:heston_crra}

Under CRRA utility, the homogeneity of $U$ suggests a separable value
function of the form
\begin{equation}\label{eq:J_separable}
  J(W, v, t) = \frac{W^{1-\gamma}}{1-\gamma} \, f(v, t),
\end{equation}
where $f(v,t) > 0$ encodes the dependence on the volatility state and the
investment horizon.  The partial derivatives are
\begin{equation}\label{eq:J_derivs}
  J_W = W^{-\gamma} f, \qquad
  J_{WW} = -\gamma W^{-\gamma-1} f, \qquad
  J_{Wv} = W^{-\gamma} f_v.
\end{equation}

The first-order condition for the portfolio weight $w$ in~\eqref{eq:hjb}
is
\begin{equation}\label{eq:foc_w}
  J_W W(\mu - \rf) + J_{WW} w v W^2 + J_{Wv} w_{\phantom{v}}
  W \kappa v \rho = 0.
\end{equation}
Substituting~\eqref{eq:J_derivs} and dividing through by
$W^{1-\gamma} f$:
\begin{equation}\label{eq:foc_w_sub}
  (\mu - \rf) - \gamma w v + \frac{f_v}{f} \kappa v \rho = 0.
\end{equation}
Solving for $w^*$:
\begin{equation}\label{eq:wstar_heston}
  \boxed{w^* = \underbrace{\frac{\mu - \rf}{\gamma v}}_{\text{myopic
  demand}} + \underbrace{\frac{f_v}{f} \cdot
  \frac{\kappa\rho}{\gamma}}_{\text{hedging demand}}}
\end{equation}

\begin{proposition}[Portfolio neutrality under stochastic volatility]
\label{prop:sv_neutrality}
Under the Heston model~\eqref{eq:heston_price}--\eqref{eq:heston_var}
with CRRA preferences~\eqref{eq:crra} and a proportional wealth tax on
all assets, the optimal portfolio weight~\eqref{eq:wstar_heston} is
independent of the wealth tax rate $\tw$.
\end{proposition}

\begin{proof}
The optimal weight $w^*$ in~\eqref{eq:wstar_heston} depends on $\mu$,
$\rf$, $\gamma$, $v$, $\kappa$, $\rho$, and the ratio $f_v/f$.  It does
not depend on the wealth level $W$.  The wealth tax $\tw$ enters the
wealth dynamics~\eqref{eq:heston_dW} and hence the HJB
equation~\eqref{eq:hjb} only through the drift of $W$.
Under the separable form~\eqref{eq:J_separable}, the function $f(v,t)$
satisfies a PDE obtained by substituting $J = W^{1-\gamma} f/(1-\gamma)$
and the optimal controls into~\eqref{eq:hjb}.  In this PDE, $\tw$ appears
only in the constant (non-$v$-dependent) term of the drift, alongside
$\rf$.  Specifically, the PDE takes the form
\begin{equation}\label{eq:f_pde}
  0 = f_t + \lambda(\theta - v) f_v + \tfrac{1}{2}\kappa^2 v f_{vv}
  + h(v;\gamma,\mu,\rf,\kappa,\rho) \, f
  + (1-\gamma)(\rf - \tw) f + g(f),
\end{equation}
where $h$ collects the $v$-dependent terms arising from the optimal
portfolio and $g(f)$ is the contribution from optimal consumption.  The
wealth tax rate $\tw$ appears only in the term $(1-\gamma)(\rf - \tw) f$,
which shifts the ``effective discount rate'' but does not interact with $v$.

Using the standard exponential-affine guess $f(v,t) = \exp(A(t) + B(t)v)$
\citep{ChackoViceira2005}, the Riccati equation for $B(t)$ arises solely
from the $v$-dependent terms in~\eqref{eq:f_pde} and is therefore
independent of $\tw$.  The ODE for $A(t)$ absorbs the constant terms
including $\tw$.  Consequently, $f_v/f = B(t)$ is independent of~$\tw$,
and the hedging demand $\frac{f_v}{f} \cdot \frac{\kappa\rho}{\gamma}$ is
tax-invariant.

Since neither the myopic demand nor the hedging demand depends on $\tw$,
the full optimal portfolio weight $w^*$ is independent of the wealth tax
rate.
\end{proof}

\begin{remark}[Interpretation]
The result has an intuitive economic interpretation.  Under CRRA,
preferences are homogeneous of degree $1-\gamma$ in wealth.  The wealth
tax scales the investor's wealth by a deterministic factor without
altering the return per unit of wealth invested or the dynamics of the
volatility state.  Since CRRA portfolio weights depend on relative
risk-return trade-offs---not on the absolute wealth level---the tax is
invisible to the portfolio decision.

This mechanism is distinct from the location-scale argument used in
\citet{Froeseth2026} for general utility.  There, neutrality follows from
a homothetic contraction of the opportunity set that preserves Sharpe
ratios.  Here, neutrality follows from the homogeneity of preferences,
which makes the opportunity set's \emph{shape} (not its scale) the sole
determinant of portfolio choice.
\end{remark}

\subsection{Consumption and welfare}\label{sec:heston_consumption}

The first-order condition for consumption in~\eqref{eq:hjb} gives
\begin{equation}\label{eq:cstar}
  C^* = W \, f(v,t)^{-1/\gamma}.
\end{equation}
The consumption-wealth ratio $C^*/W = f(v,t)^{-1/\gamma}$ depends on $v$
and $t$ but not on $W$ directly.  However, through the function $f$, the
consumption rate depends on $\tw$ via the effective discount rate in the
PDE~\eqref{eq:f_pde}: a higher tax rate reduces the level of $f$ and
hence raises the consumption-wealth ratio (the investor consumes a larger
fraction of a smaller expected wealth).  The tax therefore affects the
\emph{level} of consumption and lifetime welfare---only the \emph{portfolio
composition} is neutral.

\subsection{Generalisation to Markov diffusion models}
\label{sec:general_markov}

The argument in \Cref{sec:heston_crra} relies on two properties: (i)~the
wealth tax enters multiplicatively in the wealth dynamics and additively
in the drift, and (ii)~the CRRA value function is separable in $W$ and
the state variables.  Neither property is specific to the Heston model.

Consider a general Markov diffusion model with $K$ risky assets and $M$
state variables $\mathbf{X}_t = (X_1, \ldots, X_M)^\top$:
\begin{align}
  \frac{dS_i}{S_i} &= \mu_i(\mathbf{X}_t) \, dt
  + \sum_{j=1}^{K} \sigma_{ij}(\mathbf{X}_t) \, dW_t^{(j)},
  \qquad i = 1, \ldots, K, \label{eq:gen_price} \\[4pt]
  dX_m &= a_m(\mathbf{X}_t) \, dt
  + \sum_{j=1}^{M} b_{mj}(\mathbf{X}_t) \, dB_t^{(j)},
  \qquad m = 1, \ldots, M, \label{eq:gen_state}
\end{align}
where the Brownian motions $W^{(j)}$ and $B^{(j)}$ may be correlated.
The expected returns $\mu_i$, volatilities $\sigma_{ij}$, and state
variable dynamics $a_m$, $b_{mj}$ are all functions of the state
$\mathbf{X}_t$ but not of the investor's wealth.

The investor's wealth evolves as
\begin{equation}\label{eq:gen_dW}
  dW = \bigl\{W[\rf(\mathbf{X}_t) + \w^\top(\muvec(\mathbf{X}_t)
  - \rf(\mathbf{X}_t)\one) - \tw] - C\bigr\} \, dt
  + W \, \w^\top \Sig(\mathbf{X}_t) \, d\mathbf{W}_t.
\end{equation}

Under CRRA utility, the value function takes the form
\begin{equation}\label{eq:gen_J}
  J(W, \mathbf{X}, t) = \frac{W^{1-\gamma}}{1-\gamma}
  f(\mathbf{X}, t),
\end{equation}
and the first-order conditions for the optimal portfolio $\w^*$ yield
\begin{equation}\label{eq:gen_wstar}
  \w^* = \frac{1}{\gamma} \V(\mathbf{X})^{-1}
  \bigl(\muvec(\mathbf{X}) - \rf(\mathbf{X})\one\bigr)
  + \frac{1}{\gamma} \V(\mathbf{X})^{-1}
  \boldsymbol{\Phi}(\mathbf{X}) \,
  \frac{\nabla_{\mathbf{X}} f}{f},
\end{equation}
where $\V(\mathbf{X}) = \Sig(\mathbf{X})\Sig(\mathbf{X})^\top$ and
$\boldsymbol{\Phi}(\mathbf{X})$ is a matrix of covariances between asset
returns and state variable innovations.  Neither term depends on $W$.

\begin{proposition}[Portfolio neutrality under general Markov diffusions]
\label{prop:markov_neutrality}
Let asset returns and state variables follow the Markov
diffusion~\eqref{eq:gen_price}--\eqref{eq:gen_state}.  Under CRRA
preferences and a proportional wealth tax on all assets, the optimal
portfolio weights~\eqref{eq:gen_wstar}---including all intertemporal
hedging demands---are independent of the wealth tax rate $\tw$.

This encompasses, as special cases: geometric Brownian motion (constant
$\mu$, $\sigma$), the Heston stochastic volatility model, the Hull--White
model, the SABR model, affine term structure models with stochastic
interest rates, and any other model in which returns follow a Markov
diffusion with state-dependent coefficients.
\end{proposition}

\begin{proof}
By the same argument as \Cref{prop:sv_neutrality}.  Under the separable
value function~\eqref{eq:gen_J}, the first-order conditions for
$\w^*$ involve only the partials $J_W$, $J_{WW}$, and $J_{W X_m}$.
The ratios $J_W / (W J_{WW}) = -1/\gamma$ and
$J_{W X_m} / (W J_{WW}) = -f_{X_m}/(\gamma f)$ are both independent
of~$W$.  Since $\tw$ enters the wealth dynamics only through the drift
of~$W$, and the PDE for $f(\mathbf{X},t)$ inherits $\tw$ only in terms
that are independent of $\mathbf{X}$, the gradient ratio
$\nabla_{\mathbf{X}} f / f$ is independent of $\tw$.
\end{proof}

\subsection{When does neutrality break?}\label{sec:sv_breaks}

The stochastic volatility extension identifies CRRA as the operative
condition for neutrality beyond the location-scale family.  Neutrality
can fail when:

\begin{enumerate}
  \item \textbf{Non-CRRA preferences.}  If relative risk aversion depends
  on the wealth level---as with CARA utility, habit formation, or
  reference-dependent preferences---the portfolio weight becomes a
  function of $W$.  The wealth tax, by reducing $W$, then shifts the
  investor's risk tolerance and alters the optimal portfolio.  Under
  stochastic volatility, this effect is amplified: the hedging demand
  depends on the curvature of the value function, which is no longer
  homogeneous.

  \item \textbf{Wealth-dependent state dynamics.}  If the state variable
  dynamics depend on the investor's wealth (e.g., through market impact or
  endogenous volatility), the tax could alter the state evolution itself.
  This is empirically unlikely for a single investor but could matter in
  general equilibrium with a representative agent.

  \item \textbf{Non-universal taxation.}  If the tax applies to some
  assets but not others, the myopic demand is distorted (as in
  \citealt{Froeseth2026}, Section~9).  Under stochastic volatility, the
  hedging demand may also be distorted if the tax-exempt asset serves as a
  volatility hedge.
\end{enumerate}

\subsection{Time-scale dependence of stylised facts}\label{sec:timescale}

The stylised facts that motivate the stochastic volatility
extension---fat tails, volatility clustering, leverage effects---are not
equally pronounced at all horizons.  \Cref{app:timescale} reviews the
evidence in detail.  The key conclusion is that the most dramatic
departures from GBM are high-frequency phenomena that attenuate at
policy-relevant (monthly to annual) horizons.  What \emph{does} persist
is time-varying expected returns, regime-like volatility dynamics, and
the variance risk premium---precisely the features captured by the
general Markov diffusion framework of \Cref{sec:general_markov}.  The
stochastic volatility extension is therefore valuable because it
demonstrates that neutrality extends to these empirically relevant
features.

\section{Beyond CRRA Preferences}\label{sec:beyond_crra}

The stochastic volatility extension identifies CRRA preferences as the
operative condition for portfolio neutrality beyond the location-scale
family.  Under CRRA, the value function is homogeneous of degree
$1-\gamma$ in wealth, which makes all portfolio demands independent of
the wealth level and hence of the tax.  This section asks: what happens
when preferences depart from CRRA?

The question is empirically motivated.  Households with substantial
financial wealth often display risk-taking behaviour that is difficult to
reconcile with constant relative risk aversion.
\citet{Carroll2002} documents that the ultra-wealthy maintain saving
rates far exceeding what lifecycle consumption smoothing would predict,
suggesting that wealth itself---not just the consumption it
finances---enters the objective.  \citet{WachterYogo2010} show that
household portfolio shares in risky assets \emph{rise} with wealth, a
pattern consistent with decreasing relative risk aversion (DRRA) but
inconsistent with CRRA.  On the other hand, \citet{BrunnermeierNagel2008}
find that risky portfolio shares respond only weakly to wealth
fluctuations, a result more consistent with CRRA or with substantial
portfolio inertia.

We take the HARA (Hyperbolic Absolute Risk Aversion) utility class as
the tractable leading example.  HARA nests CRRA as a special case and
generates wealth-dependent portfolio demands through a subsistence
consumption parameter.  We derive the optimal portfolio under a wealth
tax and show that neutrality fails: the tax distorts portfolio
composition by raising the present value of subsistence needs.

\subsection{HARA utility}\label{sec:hara_setup}

The HARA class, analysed by \citet{Merton1971}, is defined by the
property that the Arrow--Pratt risk tolerance is linear in consumption:
\begin{equation}\label{eq:risk_tolerance}
  T(c) \equiv -\frac{u'(c)}{u''(c)} = \frac{c - \zeta}{\gamma},
\end{equation}
where $\zeta \geq 0$ is a subsistence (or habit) level and $\gamma > 0$
is the curvature parameter.  The corresponding utility function is
\begin{equation}\label{eq:hara_u}
  u(c) = \frac{(c - \zeta)^{1-\gamma}}{1 - \gamma}, \qquad
  c > \zeta, \quad \gamma \neq 1.
\end{equation}
When $\zeta = 0$, this reduces to CRRA.  When $\zeta > 0$, the investor
requires a minimum consumption flow $\zeta$ per unit time; only
consumption in excess of $\zeta$ generates the usual isoelastic utility.

The relative risk aversion of $u$ at consumption level $c$ is
$\gamma c/(c - \zeta)$, which exceeds $\gamma$ for all $c > \zeta$ and
converges to $\gamma$ as $c \to \infty$.  Thus HARA with $\zeta > 0$
generates \emph{decreasing relative risk aversion}: wealthier investors,
who consume further above subsistence, are effectively less risk-averse.

\subsection{Optimal portfolio with a wealth tax}\label{sec:hara_portfolio}

Consider the standard Merton problem under GBM with a single risky asset
(drift $\mu$, volatility $\sigma$), a risk-free rate $\rf$, and a
proportional wealth tax $\tw$.  The investor maximises expected
discounted lifetime utility subject to the wealth
dynamics~\eqref{eq:heston_dW} (with constant $\sigma$).

Under HARA utility~\eqref{eq:hara_u}, the value function takes the form
\begin{equation}\label{eq:hara_J}
  J(W, t) = \frac{(W - H)^{1-\gamma}}{1-\gamma} \, g(t),
\end{equation}
where $H$ is the \emph{floor wealth}---the present value of the
subsistence consumption stream $\zeta$ funded entirely by the risk-free
asset at the after-tax return $\rf - \tw$:
\begin{equation}\label{eq:floor_wealth}
  H(\tw) = \frac{\zeta}{\rf - \tw}, \qquad \rf > \tw.
\end{equation}
The condition $\rf > \tw$ is required for $H$ to be finite: if the
wealth tax exceeds the risk-free rate, the investor cannot fund
subsistence indefinitely from riskless savings alone.

The investor effectively decomposes wealth into two components: a
riskless annuity worth $H$ that funds subsistence, and a \emph{surplus}
$S = W - H$ that is invested optimally.  The surplus behaves as the
wealth of a CRRA investor.

The first-order condition for the portfolio weight in~\eqref{eq:hara_J}
gives the partial derivatives
\begin{equation}
  J_W = (W - H)^{-\gamma} g, \qquad
  J_{WW} = -\gamma (W - H)^{-\gamma - 1} g.
\end{equation}
The FOC $J_W W(\mu - \rf) + J_{WW} w \sigma^2 W^2 = 0$ becomes
\begin{equation}
  (W - H)(\mu - \rf) - \gamma w \sigma^2 W = 0,
\end{equation}
yielding the optimal portfolio weight (fraction of total wealth in the
risky asset):
\begin{equation}\label{eq:wstar_hara}
  \boxed{w^* = \frac{\mu - \rf}{\gamma \sigma^2} \cdot
  \frac{W - H(\tw)}{W}
  = w^*_{\mathrm{CRRA}} \cdot
  \left(1 - \frac{H(\tw)}{W}\right)}
\end{equation}

\subsection{Portfolio distortion}\label{sec:hara_distortion}

\begin{proposition}[Portfolio distortion under HARA preferences]
\label{prop:hara_distortion}
Under HARA utility~\eqref{eq:hara_u} with subsistence level $\zeta > 0$
and a proportional wealth tax $\tw$ on all assets, the optimal portfolio
weight~\eqref{eq:wstar_hara} is strictly decreasing in~$\tw$ for all
$W > H(\tw)$.  The tax distortion relative to the zero-tax case is
\begin{equation}\label{eq:tax_distortion}
  \Delta w^* \equiv w^*(\tw) - w^*(0)
  = -\frac{\mu - \rf}{\gamma \sigma^2} \cdot
  \frac{\zeta \, \tw}{W \, \rf (\rf - \tw)} < 0.
\end{equation}
The magnitude of the distortion is larger for investors with
(i)~higher subsistence needs $\zeta$, (ii)~lower wealth $W$,
(iii)~lower risk-free rates $\rf$, and (iv)~higher existing tax rates
$\tw$.  The distortion vanishes as $W \to \infty$ or $\zeta \to 0$.
\end{proposition}

\begin{proof}
From~\eqref{eq:wstar_hara}, $w^* = w^*_{\mathrm{CRRA}} \cdot
(1 - H(\tw)/W)$.  Since $w^*_{\mathrm{CRRA}} = (\mu - \rf)/(\gamma
\sigma^2)$ is independent of $\tw$ and $W$, the tax dependence arises
entirely through $H(\tw) = \zeta/(\rf - \tw)$.  Differentiating:
\begin{equation}\label{eq:dwdtau}
  \frac{\partial w^*}{\partial \tw}
  = -\frac{w^*_{\mathrm{CRRA}}}{W} \cdot \frac{\partial H}{\partial \tw}
  = -\frac{w^*_{\mathrm{CRRA}}}{W} \cdot
  \frac{\zeta}{(\rf - \tw)^2} < 0.
\end{equation}
The distortion~\eqref{eq:tax_distortion} follows from
$\Delta w^* = -w^*_{\mathrm{CRRA}} \cdot (H(\tw) - H(0))/W$ with
$H(\tw) - H(0) = \zeta \tw / (\rf(\rf - \tw))$.
\end{proof}

Figure~\ref{fig:hara_distortion} illustrates the mechanism.  Under
CRRA, the optimal risky share is constant across wealth levels (dashed
line).  Under HARA, the risky share increases with wealth and approaches
the CRRA value from below.  The wealth tax amplifies the distortion by
raising the floor wealth $H$, compressing the surplus $W - H$ available
for risk-taking.  The shaded region shows the additional distortion
attributable to the tax.  Two features are immediate: the distortion is
largest for investors near the subsistence floor and vanishes for the
ultra-wealthy; and each curve is defined only for $W > H(\tw)$, so
higher tax rates shrink the domain of feasible portfolios.

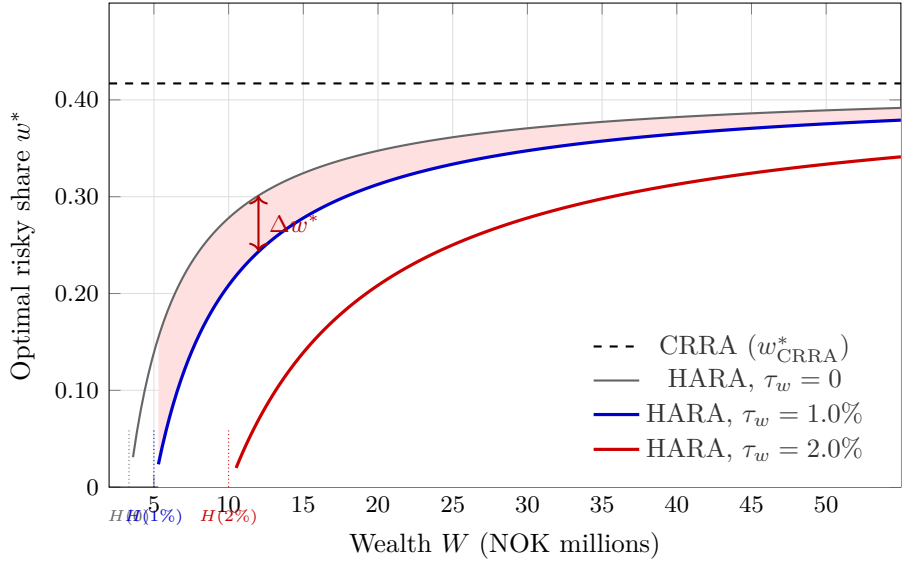
\begin{figure}[t]
\centering
\begin{tikzpicture}[scale=1.0]
  \begin{axis}[
    width=12cm, height=8cm,
    xlabel={Wealth $W$ (NOK millions)},
    ylabel={Optimal risky share $w^*$},
    xmin=2, xmax=55,
    ymin=0, ymax=0.50,
    xtick={5,10,15,20,25,30,35,40,45,50},
    ytick={0,0.10,0.20,0.30,0.40},
    yticklabels={0,0.10,0.20,0.30,0.40},
    grid=major,
    grid style={gray!25},
    legend pos=south east,
    legend style={font=\small, draw=none, fill=white, fill opacity=0.8},
    every axis label/.style={font=\small},
    tick label style={font=\footnotesize},
    clip=false,
  ]

    \addplot[fill=red!12, draw=none, domain=5.3:55, samples=200,
      forget plot]
      {0.417*(1 - 3.333/x)} \closedcycle;
    \addplot[fill=white, draw=none, domain=5.3:55, samples=200,
      forget plot]
      {0.417*(1 - 5/x)} \closedcycle;

    \addplot[black, dashed, thick, domain=2:55, samples=2]
      {0.417};
    \addlegendentry{CRRA ($w^*_{\mathrm{CRRA}}$)}

    \addplot[black!60, solid, thick, domain=3.6:55, samples=200]
      {0.417*(1 - 3.333/x)};
    \addlegendentry{HARA, $\tau_w = 0$}

    \addplot[blue!80!black, solid, very thick, domain=5.3:55, samples=200]
      {0.417*(1 - 5/x)};
    \addlegendentry{HARA, $\tau_w = 1.0\%$}

    \addplot[red!80!black, solid, very thick, domain=10.5:55, samples=200]
      {0.417*(1 - 10/x)};
    \addlegendentry{HARA, $\tau_w = 2.0\%$}

    \draw[<->, red!70!black, thick] (axis cs:12,{0.417*(1-3.333/12)})
      -- node[right, font=\footnotesize, text=red!70!black]
      {$\Delta w^*$} (axis cs:12,{0.417*(1-5/12)});

    \draw[black!60, densely dotted] (axis cs:3.333,0)
      -- (axis cs:3.333,0.06);
    \node[below, font=\tiny, text=black!60] at (axis cs:3.333,-0.012)
      {$H(0)$};
    \draw[blue!80!black, densely dotted] (axis cs:5,0)
      -- (axis cs:5,0.06);
    \node[below, font=\tiny, text=blue!80!black] at (axis cs:5,-0.012)
      {$H(1\%)$};
    \draw[red!80!black, densely dotted] (axis cs:10,0)
      -- (axis cs:10,0.06);
    \node[below, font=\tiny, text=red!80!black] at (axis cs:10,-0.012)
      {$H(2\%)$};

  \end{axis}
\end{tikzpicture}
\caption{Portfolio distortion under HARA preferences.  The dashed line
shows the CRRA benchmark (constant across wealth).  Solid curves show
the optimal risky share under HARA utility for three wealth tax rates.
The shaded region between the $\tau_w = 0$ and $\tau_w = 1\%$ curves
is the tax-induced distortion $\Delta w^*$.  Each curve begins at the
floor wealth $H(\tau_w) = \zeta / (r_f - \tau_w)$; higher taxes raise
$H$ and shrink the feasible domain.  Calibration:
$\mu - r_f = 5\%$, $\sigma = 20\%$, $\gamma = 3$,
$\zeta = 100{,}000$~NOK, $r_f = 3\%$.}
\label{fig:hara_distortion}
\end{figure}

\begin{remark}[Effective risk aversion]
\label{rem:fiscal_rra}
The result can be restated in terms of an effective relative risk
aversion.  Under the HARA value function~\eqref{eq:hara_J}, the
relative risk aversion with respect to total wealth is
\begin{equation}\label{eq:rra_eff}
  \mathrm{RRA}_{\mathrm{eff}}(W)
  = -\frac{W J_{WW}}{J_W}
  = \frac{\gamma W}{W - H(\tw)}
  = \frac{\gamma}{1 - H(\tw)/W}.
\end{equation}
This exceeds $\gamma$ whenever $H > 0$ and is increasing in~$\tw$.  The
wealth tax acts as a ``fiscal amplifier'' of risk aversion: by raising
the floor wealth $H$, it compresses the surplus $W - H$ and pushes the
investor toward more conservative portfolios.  The amplification is
strongest for investors whose wealth is close to the subsistence floor.
\end{remark}

\begin{remark}[Dollar amount vs.\ portfolio share]
While the portfolio \emph{share} $w^*$ depends on~$\tw$, the dollar
amount invested in the risky asset is
$w^* W = w^*_{\mathrm{CRRA}} \cdot (W - H)$---the CRRA-optimal
fraction of the surplus.  The tax reduces the risky-asset dollar holding
through two channels: it reduces $W$ (direct wealth drain) and increases
$H$ (higher subsistence cost), both of which shrink $W - H$.
\end{remark}

\subsection{Broader preference specifications}\label{sec:broader_prefs}

The HARA analysis illustrates a general principle: portfolio neutrality
requires preferences to be \emph{homothetic in wealth}---that is, the
value function must be homogeneous of some degree in $W$, so that
portfolio weights are scale-free.  Several empirically motivated
preference classes violate this condition.

\paragraph{Wealth in the utility function.}
\citet{BakshiChen1996} propose a ``spirit of capitalism'' model in which
utility depends on both consumption and wealth: $U = U(C, W)$.  If
wealth enters because it confers social status, power, or security
beyond its consumption value, then the marginal value of wealth is not
solely determined by the consumption it can finance.  A wealth tax in
this framework directly reduces the argument of utility, not merely the
budget constraint.  The optimal portfolio depends on the cross-partial
$U_{CW}$, which generically makes the portfolio weight a function of $W$
and hence of~$\tw$.  \citet{Carroll2002} provides empirical support:
the saving rates of the ultra-wealthy are far too high to be explained by
consumption smoothing alone, consistent with wealth entering utility
directly.

\paragraph{Epstein--Zin recursive utility.}
\citet{EpsteinZin1989} separate the coefficient of relative risk
aversion $\gamma$ from the elasticity of intertemporal substitution
$\psi$, which are constrained to satisfy $\psi = 1/\gamma$ under CRRA.

\begin{proposition}[Neutrality under Epstein--Zin]\label{prop:ez}
Under Epstein--Zin preferences with risk aversion $\gamma$ and EIS
$\psi \neq 1/\gamma$, the value function remains homogeneous of degree
$1 - \gamma$ in $W$.  The optimal portfolio weights are therefore
independent of $W$, and a proportional wealth tax is portfolio-neutral.
\end{proposition}

\begin{proof}
The Epstein--Zin aggregator over consumption and continuation value is
homogeneous of degree one in $(C, V)$.  Under homothetic
budget dynamics (as with a proportional wealth tax that scales all
returns uniformly), the value function inherits the homogeneity
of the CRRA kernel: $J(W, \mathbf{X}) = W^{1-\gamma}\,f(\mathbf{X})$.
The first-order condition for portfolio weights depends only on
$\mathbf{X}$, not on $W$.
\end{proof}

This provides a sharp boundary: departing from CRRA in the direction of
Epstein--Zin preserves neutrality (the separation of $\gamma$ and
$\psi$ matters for the consumption response and welfare cost, but not
for portfolio composition), while departing in the direction of HARA or
wealth-in-utility breaks it.

\paragraph{Non-homothetic consumption.}
\citet{WachterYogo2010} model households as consuming both necessities
and luxuries, with income elasticities that differ across goods.  As
wealth rises, a smaller fraction is devoted to necessities and the
marginal dollar is allocated to luxury consumption, which the household
is more willing to put at risk.  This generates DRRA and
wealth-dependent portfolio shares through a mechanism related to---but
distinct from---the HARA subsistence parameter.  A wealth tax in the
Wachter--Yogo framework would shift the consumption mix toward
necessities and reduce the risky portfolio share, qualitatively
consistent with the HARA result.

\subsection{Empirical evidence on wealth-dependent risk-taking}
\label{sec:drra_evidence}

The magnitude of the portfolio distortion under non-CRRA preferences
depends on how strongly risk-taking responds to wealth in practice.  The
empirical evidence is mixed but informative.

\citet{BrunnermeierNagel2008} use panel data on US household portfolios
(PSID) and find that the risky asset share responds only weakly to
wealth fluctuations driven by income shocks.  Their estimates are
consistent with CRRA or with DRRA combined with substantial portfolio
inertia.  If CRRA is approximately correct, the distortions derived above
are small.

\citet{CalvetSodini2014}, using comprehensive Swedish registry data,
estimate a financial wealth elasticity of the risky share of
approximately $0.2$: a 10\% increase in financial wealth raises the
risky share by about 2 percentage points.  This suggests moderate
DRRA.  \citet{WachterYogo2010} obtain a qualitatively similar
pattern---portfolio shares rising in wealth---when calibrating their
non-homothetic life-cycle model to US data, with an empirical
semi-elasticity of approximately 3.6 percentage points per log-unit
of net worth.

These estimates can be used to gauge the HARA distortion.  The wealth
elasticity of $w^*$ in~\eqref{eq:wstar_hara} is
\begin{equation}\label{eq:elasticity}
  \varepsilon \equiv \frac{\partial \ln w^*}{\partial \ln W}
  = \frac{H(\tw)}{W - H(\tw)},
\end{equation}
which is positive (confirming DRRA) and decreasing in $W$.  Matching
$\varepsilon \approx 0.2$ at a representative wealth level gives an
estimate of $H/W$ and hence the subsistence-to-wealth ratio that
parameterises the distortion.  At $\varepsilon = 0.2$, the implied ratio
is $H/W \approx 0.17$, meaning roughly one-sixth of wealth is committed
to subsistence.  The portfolio distortion from a 1\% wealth tax at this
calibration would be modest but nonzero---reducing $w^*$ by
approximately $0.2 \times \tw / \rf$ in relative terms.

\paragraph{Implications for policy.}
The key insight is that the portfolio distortion under non-CRRA
preferences is \emph{regressive in wealth}: it is largest for investors
near the subsistence floor and vanishes for the ultra-wealthy.  This
creates a tension with the distributional objectives of wealth taxation.
The investors most affected by the portfolio distortion are not the
ultra-rich (who are well approximated by CRRA) but those in the lower
tail of the wealth distribution subject to the tax---precisely the group
for whom the tax is least intended.  The ultra-wealthy, whose behaviour
motivates much of the policy debate \citep{Carroll2002}, are
paradoxically the group for whom the CRRA neutrality result is most
accurate.

\section{Non-Uniform Taxation}\label{sec:nonuniform}

Sections~\ref{sec:stochvol} and~\ref{sec:beyond_crra} examined
conditions under which neutrality holds or fails as a property of the
return dynamics and the investor's preferences.  We now turn to a
different source of non-neutrality: the institutional design of the
tax itself.  This and the following sections analyse features of
real-world wealth taxes---non-uniform assessment, market inelasticity,
progressive thresholds, and endogenous effort---that break neutrality
even when the investor has CRRA preferences and returns follow a
location-scale distribution.

\medskip
The neutrality results in \citet{Froeseth2026} and in
\Cref{sec:stochvol} above require the wealth tax to apply
\emph{uniformly} to all assets: every dollar of wealth is taxed at the
same rate regardless of the asset class in which it is held.  In
practice, no country satisfies this condition.  Norway and Switzerland,
the two OECD countries with active net wealth taxes, both apply
asset-class-specific valuation discounts (Norwegian:
\emph{verdsettingsrabatt}) that create differential effective tax rates
across asset classes \citep{OECD2018,Skatteetaten2026}.

This section derives the portfolio distortion that arises when the
wealth tax is non-uniform.  The result is a closed-form expression for
the tilt in optimal portfolio weights as a function of the assessment
differentials, and it connects directly to the empirical findings of
\citet{Ring2024} and \citet{FagerengGuisoRing2024}.

\subsection{Institutional setting: the Norwegian wealth tax}
\label{sec:norway_system}

Norway levies an annual net wealth tax (\emph{formuesskatt}) at a
combined municipal and state rate of 1.0\% on net wealth exceeding NOK
1.9 million (2026), rising to 1.1\% above NOK 21.5 million.  The
defining feature of the system is that different asset classes are
assessed at different fractions of their market value.  As of 2026, the
principal assessment fractions are \citep{Skatteetaten2026}:

\begin{table}[h]
\centering
\caption{Norwegian wealth tax assessment fractions (2026)}
\label{tab:norway_alpha}
\begin{tabular}{lcc}
\toprule
Asset class & Assessment fraction $\alpha_i$ & Effective discount \\
\midrule
Bank deposits          & 1.00 & 0\% \\
Secondary housing      & 1.00 & 0\% \\
Listed shares          & 0.80 & 20\% \\
Unlisted shares        & 0.80 & 20\% \\
Commercial property    & 0.80 & 20\% \\
Holiday homes          & 0.30 & 70\% \\
Primary housing ($\leq$ NOK 10m)  & 0.25 & 75\% \\
Primary housing ($>$ NOK 10m)     & 0.70 & 30\% \\
\bottomrule
\end{tabular}
\end{table}

\noindent The assessment fractions have changed substantially over
time.  Before 2023, the discount on shares and commercial property was
45\%; it was tightened to 20\% for 2023 onwards.  The housing assessment model was
overhauled in 2010, replacing historical production-cost estimates with
hedonic regression models based on market transactions
\citep{Ring2024}.  The 75\% discount on primary housing was introduced
simultaneously to offset the resulting increase in assessed values, a
political compromise that embedded a large pro-homeownership tilt into
the tax base.

\subsection{Portfolio problem with asset-class-specific assessment}
\label{sec:nonuniform_model}

Consider $K$ risky assets and one risk-free asset (e.g., bank deposits).
The risk-free asset has assessment fraction $\alpha_0$ and the risky
assets have assessment fractions $\boldsymbol{\alpha} = (\alpha_1,
\ldots, \alpha_K)^\top$.  The statutory wealth tax rate is $\tw$.

The after-tax return on the risk-free asset is $\rf - \tw \alpha_0$.
The after-tax excess return on risky asset $i$ over the risk-free asset
is
\begin{equation}\label{eq:after_tax_excess}
  (\mu_i - \tw \alpha_i) - (\rf - \tw \alpha_0)
  = (\mu_i - \rf) - \tw(\alpha_i - \alpha_0).
\end{equation}
The term $\tw(\alpha_i - \alpha_0)$ is the \emph{tax wedge}: it reduces
(increases) the effective excess return of asset $i$ when $\alpha_i >
\alpha_0$ ($\alpha_i < \alpha_0$).

The investor's wealth evolves as
\begin{equation}\label{eq:nonuniform_dW}
  dW = \bigl\{W[\rf - \tw\alpha_0
  + \w^\top(\muvec - \rf\one - \tw(\boldsymbol{\alpha} - \alpha_0\one))]
  - C\bigr\} \, dt + W \, \w^\top \Sig \, d\mathbf{W}_t.
\end{equation}

Under CRRA utility and GBM (constant $\muvec$, $\Sig$), the first-order
conditions for the optimal portfolio yield
\begin{equation}\label{eq:wstar_nonuniform}
  \boxed{\w^* = \frac{1}{\gamma} \V^{-1}
  \bigl(\muvec - \rf\one - \tw(\boldsymbol{\alpha}
  - \alpha_0\one)\bigr)}
\end{equation}
where $\V = \Sig\Sig^\top$ is the return covariance matrix.

\begin{proposition}[Portfolio distortion under non-uniform taxation]
\label{prop:nonuniform}
Under CRRA preferences and a proportional wealth tax with
asset-class-specific assessment fractions $\alpha_0, \alpha_1, \ldots,
\alpha_K$, the optimal portfolio weights~\eqref{eq:wstar_nonuniform}
depend on the tax rate $\tw$ unless all assessment fractions are equal.
The distortion relative to the uniform-tax case is
\begin{equation}\label{eq:delta_w_nonuniform}
  \Delta\w^* \equiv \w^*(\boldsymbol{\alpha}) - \w^*(\alpha_0 \one)
  = -\frac{\tw}{\gamma} \V^{-1}
  (\boldsymbol{\alpha} - \alpha_0\one).
\end{equation}
The investor overweights assets with $\alpha_i < \alpha_0$ (those
receiving a larger valuation discount than the risk-free asset) and
underweights assets with $\alpha_i > \alpha_0$.
\end{proposition}

\begin{proof}
Under uniform assessment $\alpha_i = \alpha_0$ for all $i$, the tax
enters~\eqref{eq:nonuniform_dW} as $-\tw\alpha_0 W$, a constant drain
on the drift that does not interact with $\w$.  The FOC
gives $\w^*_{\text{uniform}} = (1/\gamma)\V^{-1}(\muvec - \rf\one)$,
independent of $\tw$.  Under non-uniform assessment,
\eqref{eq:wstar_nonuniform} contains the additional term
$-(\tw/\gamma)\V^{-1}(\boldsymbol{\alpha} - \alpha_0\one)$.
Subtracting yields~\eqref{eq:delta_w_nonuniform}.
\end{proof}

\begin{remark}[Sharpe ratio distortion]
The non-uniform tax distorts after-tax Sharpe ratios.  For a single
risky asset with assessment fraction $\alpha_1$, the after-tax Sharpe
ratio is
\begin{equation}\label{eq:sr_nonuniform}
  \mathrm{SR}^{\text{after}} = \frac{\mu - \rf -
  \tw(\alpha_1 - \alpha_0)}{\sigma}.
\end{equation}
When $\alpha_1 < \alpha_0$ (as for Norwegian equities relative to bank
deposits), the after-tax Sharpe ratio \emph{exceeds} the pre-tax ratio.
The valuation discount artificially inflates the risk-adjusted return of
the tax-advantaged asset, distorting the price signal that the investor
faces.
\end{remark}

\subsection{Calibration to the Norwegian system}\label{sec:norway_calibration}

The Norwegian assessment fractions in \Cref{tab:norway_alpha} can be
used to quantify the portfolio distortion.  Consider a simplified
two-asset economy: listed equities ($\alpha_1 = 0.80$) and bank
deposits ($\alpha_0 = 1.00$).  At a statutory rate $\tw = 1.0\%$, the
tax wedge is
\begin{equation}
  \tw(\alpha_1 - \alpha_0) = 0.01 \times (0.80 - 1.00) = -0.002,
\end{equation}
i.e., the effective tax on equities is 20 basis points per year lower
than on deposits.  This is equivalent to a permanent annual subsidy of
0.2\% on the equity excess return.

For a CRRA investor with $\gamma = 4$ and equity volatility
$\sigma = 0.20$, the portfolio tilt from~\eqref{eq:delta_w_nonuniform}
is
\begin{equation}
  \Delta w^* = -\frac{0.01}{4} \cdot \frac{0.80 - 1.00}{0.04}
  = +\frac{0.002}{0.16} = +1.25\%.
\end{equation}
The investor increases the equity weight by 1.25 percentage points
relative to the neutral benchmark---a modest but non-trivial distortion
that compounds over the portfolio and across asset classes.

For primary housing ($\alpha = 0.25$), the tilt is far larger.  If
housing is modelled as an investable asset with volatility
$\sigma_h = 0.10$, the distortion becomes
\begin{equation}
  \Delta w^*_h = -\frac{0.01}{4} \cdot \frac{0.25 - 1.00}{0.01}
  = +18.75\%.
\end{equation}
This large number reflects both the 75\% discount and the low volatility
of housing.  While the simplified calculation ignores important features
of housing (illiquidity, leverage, consumption services), it illustrates
the order of magnitude: the Norwegian assessment system creates a
powerful incentive to hold primary housing relative to financial assets.

Figure~\ref{fig:portfolio_tilt} extends the two-asset calibration to the
main Norwegian asset classes, using the assessment fractions from
Table~\ref{tab:norway_alpha} and asset-class-specific volatility
estimates.  The resulting portfolio tilts span two orders of magnitude:
from zero for fully assessed assets (bank deposits, secondary housing)
to nearly 19 percentage points for primary housing below 10M NOK.  The
large spread reflects the combination of generous valuation discounts
and low asset-class volatility in the housing sector.

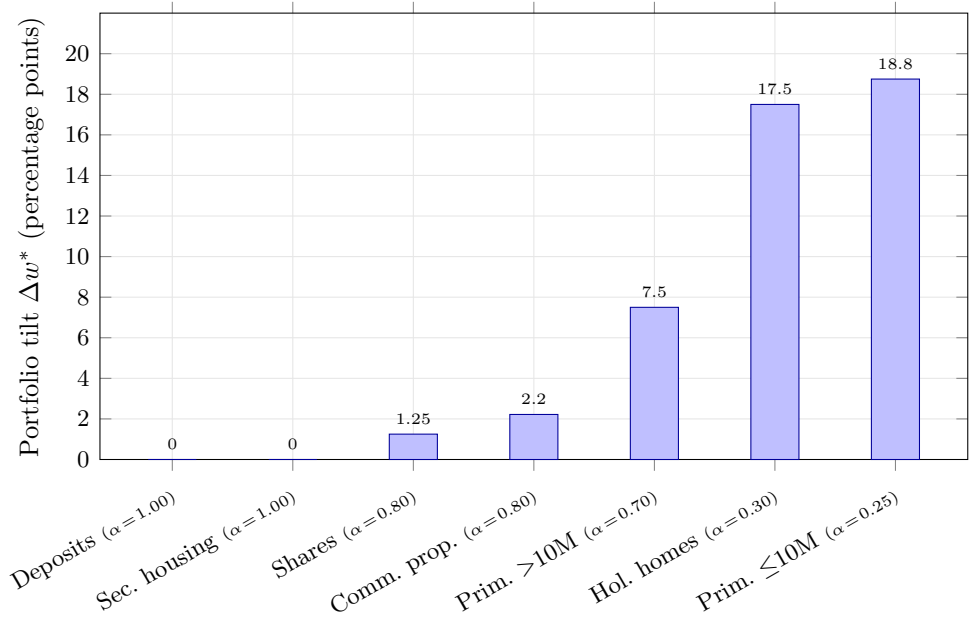
\begin{figure}[t]
\centering
\begin{tikzpicture}[scale=1.0]
  \begin{axis}[
    width=13cm, height=7.5cm,
    ybar,
    bar width=18pt,
    ylabel={Portfolio tilt $\Delta w^*$ (percentage points)},
    symbolic x coords={%
      {Deposits},
      {Sec.\ housing},
      {Shares},
      {Comm.\ prop.},
      {Prim.\ $>$10M},
      {Hol.\ homes},
      {Prim.\ $\le$10M}},
    xtick=data,
    xticklabels={%
      {Deposits {\tiny($\alpha\!=\!1.00$)}},
      {Sec.\ housing {\tiny($\alpha\!=\!1.00$)}},
      {Shares {\tiny($\alpha\!=\!0.80$)}},
      {Comm.\ prop.\ {\tiny($\alpha\!=\!0.80$)}},
      {Prim.\ $>$10M {\tiny($\alpha\!=\!0.70$)}},
      {Hol.\ homes {\tiny($\alpha\!=\!0.30$)}},
      {Prim.\ $\le$10M {\tiny($\alpha\!=\!0.25$)}}},
    x tick label style={font=\footnotesize, rotate=30, anchor=north east},
    ymin=0, ymax=22,
    ytick={0,2,4,6,8,10,12,14,16,18,20},
    grid=major,
    grid style={gray!20},
    every axis label/.style={font=\small},
    tick label style={font=\footnotesize},
    nodes near coords,
    nodes near coords style={font=\tiny, above},
    point meta=explicit symbolic,
    clip=false,
  ]
    \addplot[fill=blue!25, draw=blue!60!black] coordinates {
      ({Deposits},        0)    [0]
      ({Sec.\ housing},   0)    [0]
      ({Shares},          1.25) [1.25]
      ({Comm.\ prop.},    2.22) [2.2]
      ({Prim.\ $>$10M},   7.50) [7.5]
      ({Hol.\ homes},     17.50)[17.5]
      ({Prim.\ $\le$10M}, 18.75)[18.8]
    };

  \end{axis}
\end{tikzpicture}
\caption{Portfolio tilt $\Delta w^*$ toward each asset class relative to
bank deposits, under the Norwegian assessment system.  The tilt is
computed from~\eqref{eq:delta_w_nonuniform} with $\tau_w = 1.0\%$,
$\gamma = 4$, and asset-class volatilities $\sigma = 0.20$ (listed
shares), $\sigma = 0.15$ (commercial property), $\sigma = 0.10$
(housing classes).  Assessment fractions $\alpha_i$ are shown below each
bar.  The two-order-of-magnitude spread illustrates the strong incentive
to hold primary housing over financial assets.}
\label{fig:portfolio_tilt}
\end{figure}

\subsection{Leverage and debt deductibility}\label{sec:leverage}

The portfolio distortions derived above are amplified by the interaction
between valuation discounts and debt deductibility.  In the Norwegian
system, debt has historically been deductible at face value against the
wealth tax base, while assets receiving valuation discounts are assessed
below market value.  For a leveraged position in asset $i$ with
loan-to-value ratio $\ell_i = D_i/V_i$, the net contribution to taxable
wealth per unit of asset value is
\begin{equation}\label{eq:taxable_leverage}
  \frac{\text{Taxable wealth}}{V_i} = \alpha_i - \beta_i \, \ell_i,
\end{equation}
where $\beta_i \leq 1$ is the fraction of associated debt that is
deductible.

\paragraph{Leverage and effective tax rates.}
Under full debt deduction ($\beta_i = 1$), the net taxable contribution
$\alpha_i - \ell_i$ becomes negative whenever leverage exceeds the
assessment fraction.  This created a powerful sheltering strategy: highly
leveraged positions in discounted asset classes could offset the tax
base from other assets (see Appendix~\ref{app:leverage} for historical
examples and calibrations).  Since total net taxable wealth is floored
at zero, the mechanism operates by sheltering positive wealth in fully
assessed assets.

Norway's proportional debt reduction (\emph{gjeldsreduksjon}) sets
$\beta_i = \alpha_i$ for discounted assets, giving
\begin{equation}\label{eq:taxable_proportional}
  \frac{\text{Taxable wealth}}{V_i} = \alpha_i(1 - \ell_i),
\end{equation}
which is non-negative for all $\ell_i \leq 1$, eliminating the
negative tax base.  However, differential effective rates persist.
Under full debt deduction, the effective wealth tax rate on a leveraged
position with equity $E_i = V_i(1 - \ell_i)$ is
\begin{equation}\label{eq:effective_rate_leverage}
  \tau_{w,i}^{\text{eff}} = \tw \cdot \frac{\alpha_i - \ell_i}{1 - \ell_i},
\end{equation}
which can be negative---a qualitatively distinct distortion that creates
incentives to leverage into discounted asset classes.  The proportional
debt reduction rule attenuates this to
$\tau_{w,i}^{\text{eff}} = \tw \alpha_i$, which remains below $\tw$
for discounted assets.

\subsection{Empirical evidence}\label{sec:nonuniform_empirical}

The theoretical distortions derived above have empirical counterparts.
\Cref{app:empirical_nonuniform} reviews the evidence from Norway
\citep{Ring2024,FagerengGuisoRing2024} and Switzerland
\citep{Brulhart2022} in detail.  The key findings are: (i)~portfolio
composition is unaffected when the tax does not discriminate between
asset classes, consistent with our neutrality result under uniform
assessment; (ii)~when assessment differentials create distinct after-tax
returns, households do rebalance, but slowly---full reoptimisation takes
approximately five years; and (iii)~reported wealth responses to tax
rate changes are large but dominated by mobility and valuation responses
rather than real portfolio reallocation.  The sluggish adjustment
implies that the theoretical distortions derived in this paper
represent long-run equilibrium magnitudes; short-run portfolio
responses to a newly introduced or reformed wealth tax will be
smaller.

\subsection{Interaction with non-CRRA preferences}\label{sec:nonuniform_hara}

The distortions from non-uniform taxation and from non-CRRA preferences
identified in \Cref{sec:beyond_crra} can compound.  Under HARA utility
with subsistence level $\zeta > 0$ and asset-class-specific assessment,
the optimal weight of the single risky asset becomes
\begin{equation}\label{eq:wstar_hara_nonuniform}
  w^* = \frac{\mu - \rf - \tw(\alpha_1 - \alpha_0)}{\gamma\sigma^2}
  \cdot \frac{W - H(\tw, \alpha_0)}{W},
\end{equation}
where $H(\tw, \alpha_0) = \zeta/(\rf - \tw\alpha_0)$ is the floor
wealth evaluated at the after-tax risk-free rate.  The first factor
captures the Sharpe-ratio distortion from non-uniform assessment
(\Cref{prop:nonuniform}); the second captures the surplus-wealth effect
from non-CRRA preferences (\Cref{prop:hara_distortion}).  The two
distortions are multiplicative and can reinforce each other: an investor
near the subsistence floor holding an asset with a favourable assessment
discount experiences both an inflated effective Sharpe ratio and a
compressed risk-taking capacity.

\begin{remark}[Stochastic volatility and non-uniform assessment]
When return volatility varies across asset classes---equities being
substantially more volatile than real estate at monthly
frequencies---the hedging demand from the stochastic volatility
extension (Section~\ref{sec:stochvol}) interacts with non-uniform
assessment.  An investor who hedges against volatility shocks by tilting
toward low-volatility assets may find this tilt reinforced or opposed by
the tax-induced tilt from assessment differentials.  Formalising this
interaction requires a multi-asset stochastic volatility model with
asset-class-specific dynamics, which we leave for future work.
\end{remark}

\subsection{Policy implications}\label{sec:nonuniform_policy}

The analysis highlights a fundamental tension in wealth tax design.
Uniform assessment across all asset classes---the condition required for
portfolio neutrality---conflicts with several practical objectives.

First, valuation accuracy varies by asset class.  Bank deposits,
listed equities, and government bonds have observable market prices;
private business equity, real estate, and collectibles do not
\citep{OECD2018}.  Assessment discounts are often justified as
compensation for valuation uncertainty, but they introduce portfolio
distortions as a side effect.

Second, assessment discounts serve political objectives.  Norway's
75\% discount on primary housing reflects a policy choice to protect
homeownership, not a valuation concern---the hedonic regression model
introduced in 2010 provides accurate market values
\citep{Ring2024}.  The discount was introduced precisely to offset the
increase in assessed values that the improved model would generate.

Third, differential assessment creates an implicit industrial policy:
it channels capital toward tax-advantaged asset classes and away from
those assessed at full market value.  The Norwegian system, which taxes
bank deposits at full value while discounting housing by 75\%, provides
a large incentive to hold housing wealth---an allocation that may reduce
productive capital formation.

Fourth, as shown in \Cref{sec:leverage}, valuation discounts interact
with debt deductibility to create a leverage channel that amplifies the
portfolio distortion.  Under full debt deduction, leveraged positions in
discounted asset classes can generate negative taxable contributions,
allowing investors to shelter wealth held in other asset classes.  This
incentivises not only a tilt toward discounted assets but also excessive
leverage in those positions---a distortion with potential systemic
implications for financial stability.  Norway's proportional debt
reduction rules (\emph{gjeldsreduksjon}) eliminate the most egregious
arbitrage (negative tax bases), but the remaining differential in
effective rates continues to favour leveraged investment in discounted
asset classes over unleveraged holdings in fully assessed ones.

From the perspective of neutrality, the first-best policy is uniform
assessment at market value ($\alpha_i = 1$ for all $i$) combined with
full debt deductibility at face value.
\citet{ScheuerSlemrod2021} survey the design challenges and note that
uniform assessment is feasible for liquid assets but faces significant
obstacles for illiquid ones.  A second-best alternative is to equalise
assessment fractions across liquid financial assets (equities, bonds,
deposits) while accepting that illiquid assets require separate
treatment.  This would preserve portfolio neutrality within the
financial portfolio while acknowledging that housing and private business
equity present distinct valuation challenges.  Any remaining valuation
discounts should be paired with proportional debt reduction to prevent
the leverage arbitrage documented in \Cref{sec:leverage}.

These portfolio distortions are derived in partial equilibrium: they
take asset prices as given.  When all wealth-tax payers simultaneously
adjust their portfolios, the resulting aggregate flows feed back into
equilibrium prices.  Section~\ref{sec:inelastic} examines the magnitude
of this general equilibrium response.

\section{Inelastic Markets and General Equilibrium}\label{sec:inelastic}

The preceding sections analyse the wealth tax from the perspective of an
individual investor who takes asset returns as given.  In this partial
equilibrium setting, the CRRA neutrality result states that portfolio
\emph{weights} are unaffected by the tax.  But the tax reduces the
investor's total wealth, and hence the \emph{dollar amount} allocated to
each asset.  When all taxed investors simultaneously reduce their dollar
holdings, aggregate demand for risky assets falls.  This section asks:
how do equilibrium prices respond, and does the magnitude of the
response depend on the structure of market demand?

The answer, we argue, depends critically on the price elasticity of
aggregate equity demand.  Under the traditional assumption of elastic
markets---where many marginal investors stand ready to absorb demand
shifts at close to fundamental value---the price adjustment is modest and
efficient.  Under the \emph{inelastic markets hypothesis}
\citep{GabaixKoijen2021}, aggregate equity demand is far less elastic
than commonly assumed, and even modest flows can produce amplified price
effects.

\subsection{From partial to general equilibrium}\label{sec:pe_to_ge}

Under the CRRA neutrality result (\Cref{sec:stochvol},
\Cref{sec:general_markov}), the optimal portfolio weight $w^*$ is
independent of~$\tw$.  However, the dollar amount invested in equities
is $w^* W$, and $W$ is reduced by the tax.  In the long run, an
investor facing a perpetual tax $\tw$ accumulates less wealth by a
factor that depends on $\tw$ relative to the growth rate of wealth.

If a fraction $\phi$ of total equity market capitalisation is held by
taxed investors, the aggregate demand shift induced by the tax is
\begin{equation}\label{eq:aggregate_flow}
  \Delta F = -\tw \cdot \phi \cdot P_{\text{eq}},
\end{equation}
where $P_{\text{eq}}$ is the equilibrium market capitalisation and
$\Delta F$ represents the annual net flow out of equities required to
fund tax payments.  We treat $\Delta F$ as a flow (per period) rather
than a level shift, since the tax recurs annually.

In a frictionless general equilibrium with fully elastic demand, this
flow is absorbed efficiently: prices adjust by exactly the amount needed
to reflect the lower after-tax wealth, expected returns rise
commensurately, and the resulting equilibrium is Pareto-efficient
conditional on the tax.  The ``price impact'' is simply the new
fundamental value.

\subsection{The inelastic markets hypothesis}\label{sec:imh}

\citet{GabaixKoijen2021} challenge the elastic-demand assumption.
Using granular instrumental variables constructed from mutual fund flows,
they estimate the price elasticity of aggregate US equity demand and find
a striking result: \emph{investing \$1 in the stock market increases the
aggregate market's value by approximately \$5}.  The multiplier $M$,
defined as the ratio of the price change to the flow that caused it, is
estimated at approximately $M \approx 5$, with a range of 3--8 across
specifications.

The low elasticity arises because the marginal holders of equity---index
funds, pension funds, insurance companies---operate under mandates that
fix their equity allocations within narrow bands.  When aggregate demand
shifts, few participants can absorb the change, and prices must move
substantially to clear the market.  \citet{KoijenYogo2019} develop the
demand-system framework that underpins these estimates, showing that
institutional demand is far more price-inelastic than the efficient
markets paradigm assumes.

The micro-level foundations are provided by \citet{BouchaudFarmerLillo2009}
and \citet{Bouchaud2022}.  Large institutional orders are split into
many small child orders and executed over days to months, creating
persistent order flow that moves prices incrementally.  Revealed
liquidity---the depth of the visible order book near the current
price---is extremely low relative to typical order sizes.
\citet{Bouchaud2022} derives the aggregate multiplier from a latent
order book model: because latent liquidity is sparse near the prevailing
price, even moderate flows push through the available depth and generate
disproportionate price changes.  The empirical signature is the
well-documented square-root law of price impact, $I(Q) \propto
\sqrt{Q}$, which holds across stocks, futures, and options markets
\citep{BouchaudFarmerLillo2009}.

The fire sales literature provides further micro-level evidence.
\citet{CovalStafford2007} show that mutual funds experiencing large
outflows create measurable price pressure in their commonly held
securities, with prices subsequently reverting---consistent with
temporary demand-driven dislocations rather than information-driven
price changes.

\subsection{Price impact of wealth-tax-induced flows}
\label{sec:tax_price_impact}

Combining the neutrality framework with the inelastic markets
hypothesis yields a distinction between two effects of the wealth tax
on asset prices:

\paragraph{Fundamental valuation effect.}
The tax reduces the after-tax wealth of taxed investors, lowering their
lifetime consumption and hence the fundamental value of claims on their
future income.  In a representative-agent economy, this would reduce
equilibrium prices by a factor proportional to $\tw$ relative to the
discount rate.  This effect is present in any general equilibrium model
and does not depend on market elasticity.

\paragraph{Flow-amplification effect.}
Under inelastic markets, the aggregate flow $\Delta F$ required to fund
tax payments generates a price response that \emph{exceeds} the
fundamental valuation change.  If $\Delta F / P_{\text{eq}}$ is the flow
as a fraction of market capitalisation, the price impact is
\begin{equation}\label{eq:price_impact_imh}
  \frac{\Delta P}{P} \approx -M \cdot \frac{\Delta F}{P_{\text{eq}}}
  = -M \cdot \tw \cdot \phi,
\end{equation}
where $M$ is the Gabaix--Koijen multiplier and $\phi$ is the share of
market capitalisation held by taxed investors.  The excess price impact
$(M - 1) \cdot \tw \cdot \phi$ represents a flow-driven dislocation
that does not arise in frictionless markets.

\Cref{fig:price_wedge} illustrates the mechanism.  The supply of equity
is fixed (vertical).  A wealth tax shifts demand to the left by the
flow $\Delta F$.  Under elastic demand, the equilibrium price falls
modestly from $P_0$ to $P_1^E$; under inelastic demand, the same flow
produces a much larger decline to $P_1^I$.  The shaded region represents
the excess price impact---the amplification due to market inelasticity.

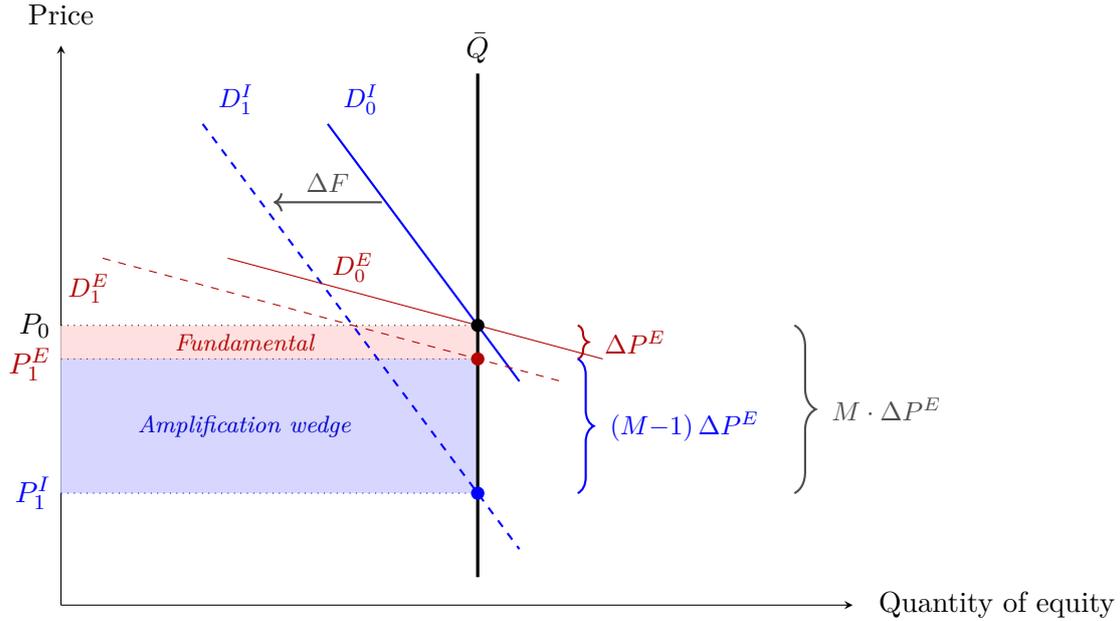
\begin{figure}[t]
\centering
\begin{tikzpicture}[scale=1.0]
  \begin{axis}[
    width=12cm, height=9cm,
    xlabel={Quantity of equity},
    ylabel={Price},
    xmin=0, xmax=9.5, ymin=0, ymax=10,
    xtick=\empty, ytick=\empty,
    axis lines=left,
    clip=false,
    every axis x label/.style={at={(ticklabel* cs:1.02)},
      anchor=west},
    every axis y label/.style={at={(ticklabel* cs:1.02)},
      anchor=south},
  ]
    \fill[red!15, opacity=0.8] (axis cs:0,4.4) rectangle (axis cs:5,5);
    \fill[blue!20, opacity=0.8] (axis cs:0,2) rectangle (axis cs:5,4.4);

    \addplot[very thick, black] coordinates {(5,0.5) (5,9.5)};
    \node[above] at (axis cs:5,9.5) {$\bar{Q}$};

    \addplot[thick, blue] coordinates {(3.2,8.6) (5.5,4)};
    \node[blue, above, font=\small] at (axis cs:3.6,8.6) {$D_0^I$};
    \addplot[thick, blue, dashed] coordinates {(1.7,8.6) (5.5,1)};
    \node[blue, above, font=\small] at (axis cs:2.1,8.6) {$D_1^I$};

    \addplot[thin, red!70!black] coordinates {(2,6.2) (6.5,4.4)};
    \node[red!70!black, above, font=\small] at (axis cs:3.5,5.6) {$D_0^E$};
    \addplot[thin, red!70!black, dashed] coordinates {(0.5,6.2) (6,4.0)};
    \node[red!70!black, below left, font=\small] at (axis cs:0.7,6.1) {$D_1^E$};

    \addplot[thin, black, dotted] coordinates {(0,5) (5,5)};
    \node[left] at (axis cs:0,5) {$P_0$};

    \addplot[thin, red!70!black, dotted] coordinates {(0,4.4) (5,4.4)};
    \node[left, red!70!black, yshift=-2pt] at (axis cs:0,4.4) {$P_1^E$};

    \addplot[thin, blue, dotted] coordinates {(0,2) (5,2)};
    \node[left, blue] at (axis cs:0,2) {$P_1^I$};

    \draw[->, thick, black!70] (axis cs:3.85,7.2) -- (axis cs:2.55,7.2);
    \node[above, black!70, font=\small] at (axis cs:3.2,7.2) {$\Delta F$};

    \draw[decorate, decoration={brace, amplitude=4pt, mirror},
      thick, red!70!black]
      (axis cs:6.2,4.4) -- (axis cs:6.2,5)
      node[midway, right=6pt, red!70!black, font=\small]
      {$\Delta P^E$};

    \draw[decorate, decoration={brace, amplitude=6pt, mirror},
      thick, blue]
      (axis cs:6.2,2) -- (axis cs:6.2,4.4)
      node[midway, right=8pt, blue, font=\small]
      {$(M{-}1)\,\Delta P^E$};

    \draw[decorate, decoration={brace, amplitude=8pt, mirror},
      thick, black!70]
      (axis cs:8.8,2) -- (axis cs:8.8,5)
      node[midway, right=10pt, black!70, font=\small]
      {$M \cdot \Delta P^E$};

    \node[red!70!black, font=\footnotesize\itshape] at (axis cs:2.2,4.7)
      {Fundamental};
    \node[blue!80!black, font=\footnotesize\itshape] at (axis cs:2.2,3.2)
      {Amplification wedge};

    \fill[black] (axis cs:5,5) circle (2.5pt);
    \fill[red!70!black] (axis cs:5,4.4) circle (2.5pt);
    \fill[blue] (axis cs:5,2) circle (2.5pt);

  \end{axis}
\end{tikzpicture}
\caption{Decomposition of the price impact of a wealth-tax-induced
demand shift.  Supply is fixed at $\bar{Q}$.  The wealth tax shifts
demand left by $\Delta F = \tau_w \phi P_0$.  The total price decline
from $P_0$ to $P_1^I$ decomposes into two components: (i)~the
\emph{fundamental effect} (light red), equal to $\Delta P^E = \tau_w
\phi$, which is the decline that would occur under perfectly elastic
demand; and (ii)~the \emph{amplification wedge} (light blue), equal
to $(M{-}1)\,\Delta P^E$, which is the excess decline due to market
inelasticity.  The total impact is $M$ times the fundamental effect,
where $M \approx 5$ is the Gabaix--Koijen multiplier.
The diagram shows the full tax liability as the demand shift;
actual equity outflows are smaller once dividend and liquid-income
channels are accounted for (see \Cref{sec:imh_calibration}).}
\label{fig:price_wedge}
\end{figure}

\begin{remark}[Incidence on non-taxed investors]
\label{rem:incidence}
The price decline falls on \emph{all} equity holders, not only those
subject to the wealth tax.  In a small open economy such as Norway,
where foreign investors hold a large share of the equity market, a
substantial fraction of the price impact is borne by non-taxed
investors.  This creates an incidence shift: the effective burden of the
wealth tax is partly ``exported'' through the price channel.
Conversely, in a closed economy or one where only domestic investors
hold equities, the price decline feeds back into the tax base, reducing
future tax revenue.
\end{remark}

\subsection{Calibration}\label{sec:imh_calibration}

A back-of-the-envelope calculation illustrates the potential magnitude.
Consider a wealth tax at rate $\tw = 1\%$ applied to equity holdings
that constitute $\phi = 20\%$ of total market capitalisation (the
remainder being held by tax-exempt institutions and foreign investors).
Under the conservative assumption that the entire tax liability is met
by selling equities (refined below), the tax-induced flow as a fraction
of market capitalisation is
\begin{equation}
  \frac{\Delta F}{P_{\text{eq}}} = \tw \cdot \phi = 0.01 \times 0.20
  = 0.002 = 0.2\%.
\end{equation}
Under frictionless markets ($M = 1$), prices decline by 0.2\%---a
negligible effect.  Under the inelastic markets hypothesis ($M = 5$),
prices decline by
\begin{equation}
  \frac{\Delta P}{P} = -5 \times 0.002 = -1.0\%.
\end{equation}
A one percent annual price decline, compounded over the holding period,
represents a meaningful reduction in wealth for all equity holders.

\Cref{tab:sensitivity_M} shows the sensitivity of this estimate to the
multiplier across the range reported in the literature.

\begin{table}[h]
\centering\small
\begin{tabular}{lccc}
\toprule
& \multicolumn{3}{c}{Multiplier $M$} \\
\cmidrule(lr){2-4}
& 3 (low) & 5 (central) & 8 (high) \\
\midrule
$\Delta P/P$ (Norwegian, $\tw = 1\%$, $\phi = 0.20$) & $-0.6\%$ & $-1.0\%$ & $-1.6\%$ \\
$\Delta P/P$ (Zucman, $\tw = 2\%$, $\phi_B = 0.08$)  & $-0.5\%$ & $-0.8\%$ & $-1.3\%$ \\
\bottomrule
\end{tabular}
\caption{Sensitivity of annual equilibrium price impact to the
Gabaix--Koijen multiplier $M$.  Central estimate $M = 5$; range 3--8
from \citet{GabaixKoijen2021}.}
\label{tab:sensitivity_M}
\end{table}

Several caveats apply.  First, the Gabaix--Koijen multiplier is
estimated for the US equity market; no published estimates exist for
global equity markets or specific European exchanges.  For the
Saez--Zucman global proposal, the relevant multiplier applies to the
world equity market.  On one hand, greater aggregate depth may lower
$M$ relative to the US estimate.  On the other, a coordinated global
tax eliminates the substitution channel through which selling pressure
in one market is absorbed by capital inflows from untaxed
jurisdictions, which could sustain or increase $M$.
\citet{KoijenYogo2019} estimate institution-level price impact
elasticities of 0.04--0.16 in their demand-system framework, confirming
that institutional demand is highly inelastic, but the mapping from
these micro elasticities to an aggregate flow multiplier is not
direct.  The net effect is ambiguous, and the sensitivity range in
\Cref{tab:sensitivity_M} should be interpreted accordingly.  Second, and most importantly, the model assumes that taxed investors
liquidate equities to meet their obligations.  In practice, investors
typically cover the wealth tax through a hierarchy of payment sources:
other liquid income (salary, interest, rental income), dividend
extractions from their companies, and---only as a last
resort---outright share sales
\citep[see][for a detailed discussion]{Froeseth2026}.
\citet{BerzinsBohrenStacescu2022} find that Norwegian private firms
increase dividend payouts in response to their owners' wealth tax
obligations, confirming that dividend extraction is the dominant
margin.  \citet{BjornebyEtAl2023} document the same pattern for
closely held firms, where dividend payments rise roughly in proportion
to the tax liability.  Using Colombian data, \citet{PeydroEtAl2025}
show that the firm-level response depends on the financial environment:
firms whose owners face the wealth tax increase leverage by
substituting bank debt for extracted equity, and those with credit
access actually increase investment---suggesting that the sign of the
real effect is ambiguous and depends on the availability of alternative
financing.  The effective equity outflow $\Delta F$ is
therefore substantially smaller than the full tax liability
$\tw\,\phi\,P_0\,Q$, and should be interpreted as the residual selling
pressure after dividends and liquid assets are exhausted.

Furthermore, in Norway the personal wealth tax is collected through
quarterly advance payments (\emph{forskuddsskatt}) during the
assessment year, not as a single year-end lump sum.  Whatever equity
selling does occur is therefore spread over four quarters, reducing the
instantaneous flow impact and giving markets time to absorb the
pressure gradually.
Third, a pre-announced, recurring wealth tax allows markets to
anticipate the flow, potentially front-loading the price adjustment.
Fourth, the square-root impact law $I(Q) \propto \sqrt{Q}$ suggests
that large aggregate flows may have less-than-proportional impact per
unit, moderating the linear approximation
in~\eqref{eq:price_impact_imh}.

Norway introduced a deferral scheme (\emph{utsettelsesordning}) for
wealth tax payments starting in 2026, allowing taxpayers to defer
payment for up to three years at a market-rate interest charge.  This
smooths the liquidity demand and may reduce the flow-amplification
effect, though it does not eliminate the fundamental valuation effect.

\subsection{Implications for neutrality}\label{sec:imh_neutrality}

The inelastic markets framework introduces a distinction between two
levels of neutrality that the partial equilibrium analysis conflates:

\begin{enumerate}
  \item \textbf{Portfolio-weight neutrality.}  Under CRRA preferences
  and uniform assessment, each investor's optimal portfolio
  \emph{weights} are independent of~$\tw$.  This result survives in
  general equilibrium: even after prices adjust, the optimal weight
  remains $w^* = (\mu' - \rf)/(\gamma\sigma'^2)$ evaluated at the new
  equilibrium $\mu'$, $\sigma'$, and the investor's portfolio share is
  unchanged.

  \item \textbf{Price neutrality.}  The tax is \emph{not} neutral with
  respect to equilibrium asset prices.  Under inelastic markets, the
  price impact is amplified by the multiplier $M$, creating a wedge
  between the pre-tax and post-tax equilibrium that exceeds the
  fundamental valuation change.
\end{enumerate}

The distinction matters for policy evaluation.  Portfolio-weight
neutrality means the tax does not distort relative asset allocation---an
important efficiency property.  But the lack of price neutrality means
the tax can depress asset values, reduce market capitalisation, and
shift incidence to non-taxed investors.  These general equilibrium
effects operate through a channel entirely absent from the partial
equilibrium analysis that dominates the wealth tax literature.

The interaction with the extensions developed earlier in this paper
creates additional channels.  Under non-CRRA preferences
(\Cref{sec:beyond_crra}), the price decline reduces investors' surplus
wealth $W - H$, increasing effective risk aversion and potentially
triggering further selling---a feedback loop between the
flow-amplification effect and the HARA distortion.  Under non-uniform
assessment (\Cref{sec:nonuniform}), asset-class-specific flows generate
differential price impacts across markets, with larger effects in
markets for illiquid or concentrated-ownership assets where the
effective multiplier may substantially exceed the aggregate estimate.

A full general equilibrium analysis integrating the Merton portfolio
framework with the Gabaix--Koijen demand system is beyond the scope of
this paper but represents an important direction for future work.

\section{Progressive Taxation and Threshold Effects}\label{sec:progressive}

The neutrality result of \citet{Froeseth2026} rests on the tax being
\emph{proportional}: the same rate $\tau_w$ applied to every unit of
wealth.  Every wealth tax implemented in practice departs from
proportionality by imposing a positive threshold (the Norwegian
\emph{bunnfradrag}) below which no tax is due, and often multiple
brackets above it.  This section formalises the portfolio distortion
created by progressive taxation and calibrates it to the Norwegian
system.

\subsection{General framework}\label{sec:prog_framework}

Consider a progressive wealth tax with $K$ brackets.  Let
$\bar{W}_1 < \bar{W}_2 < \cdots < \bar{W}_K$ be the bracket thresholds
and $\tau_1 < \tau_2 < \cdots < \tau_K$ the associated marginal rates,
with $\tau_0 = 0$ below $\bar{W}_1$.  The tax liability for an investor
with wealth $W$ in bracket $j$ (i.e.\ $\bar{W}_j < W \le \bar{W}_{j+1}$,
with $\bar{W}_{K+1} = \infty$) is
\begin{equation}\label{eq:prog_tax}
  T(W) \;=\; \sum_{k=1}^{j} (\tau_k - \tau_{k-1})\,
         \max\!\bigl(0,\, W - \bar{W}_k\bigr)
       \;=\; \tau_j\, W \;-\; R_j\,,
\end{equation}
where the \emph{cumulative rebate}
\begin{equation}\label{eq:rebate}
  R_j \;=\; \sum_{k=1}^{j} (\tau_k - \tau_{k-1})\,\bar{W}_k
\end{equation}
is the tax saving, relative to a proportional tax at the marginal rate
$\tau_j$, generated by the exemption and lower-rate brackets.  The
decomposition in~\eqref{eq:prog_tax} is the key observation: a
progressive tax equals a proportional tax at the marginal rate minus a
lump-sum rebate that depends only on the bracket structure, not on the
portfolio.

The effective average rate is
\begin{equation}\label{eq:avg_rate}
  \bar{\tau}(W) \;=\; \frac{T(W)}{W} \;=\; \tau_j - \frac{R_j}{W}\,,
\end{equation}
which is strictly increasing in $W$ and converges to the marginal rate
$\tau_j$ as $W \to \infty$.

\subsection{Portfolio distortion under CRRA}
\label{sec:prog_portfolio}

For an investor in bracket $j$ with $W \gg \bar{W}_j$, the wealth
dynamics are
\begin{equation}\label{eq:prog_dynamics}
  dW = \bigl[(r_f - \tau_j)\,W + w\,W\,(\mu - r_f)
       + R_j - c\bigr]\,dt + w\,W\,\sigma\,dZ\,.
\end{equation}
The proportional component $\tau_j\,W$ reduces all returns by $\tau_j$,
preserving neutrality of the excess return $\mu - r_f$.  The rebate
$R_j$ enters as a constant income flow, independent of the portfolio.
In the Merton continuous-time framework, a constant income flow $y$ is
equivalent to a riskless asset with present value $y/(r_f^a)$, where
$r_f^a = r_f - \tau_j$ is the after-tax risk-free rate
\citep{Merton1971}.

\begin{definition}[Tax shield]\label{def:tax_shield}
The \emph{tax shield} of a progressive wealth tax for an investor in
bracket $j$ is the present value of the cumulative rebate:
\begin{equation}\label{eq:tax_shield}
  H_\tau \;=\; \frac{R_j}{r_f - \tau_j}\,.
\end{equation}
\end{definition}

Under CRRA preferences, the optimal risky share of financial wealth is
determined by the ratio of total wealth (financial plus implicit) to
financial wealth alone.

\begin{proposition}[Progressive tax distortion]\label{prop:progressive}
Under a progressive wealth tax with bracket structure
$\{(\bar{W}_k, \tau_k)\}_{k=1}^K$, an investor with CRRA preferences,
risk aversion $\gamma$, and financial wealth $W$ in bracket $j$ holds
a risky portfolio share
\begin{equation}\label{eq:prog_wstar}
  w^* \;=\; \frac{\mu - r_f}{\gamma\,\sigma^2}\;
            \frac{W + H_\tau}{W}
      \;=\; w^*_{\mathrm{neutral}}\;\Bigl(1 + \frac{H_\tau}{W}\Bigr)\,,
\end{equation}
where $w^*_{\mathrm{neutral}} = (\mu - r_f)/(\gamma\sigma^2)$ is the
proportional-tax optimal share and $H_\tau$ is the tax shield
from~\eqref{eq:tax_shield}.
\end{proposition}

\begin{proof}
With the constant income flow $R_j$, the investor's total investable
wealth is $W + H_\tau$, of which $H_\tau$ is implicitly riskless
(the rebate is earned regardless of portfolio returns, conditional on
$W > \bar{W}_j$).  The CRRA optimal risky share of total wealth is
$w^*_{\mathrm{neutral}}$; expressing this as a share of financial
wealth $W$ gives~\eqref{eq:prog_wstar}.
\end{proof}

\Cref{prop:progressive} establishes three results.  First, the
progressive structure \emph{increases} the risky share relative to the
proportional benchmark.  The tax shield $H_\tau$ acts like a riskless
endowment---the exempted wealth generates a certain tax saving each
period---which tilts the portfolio toward risk.  Second, the distortion
is \emph{decreasing} in wealth:
\[
  \frac{\partial}{\partial W}\,\frac{H_\tau}{W}
  \;=\; -\frac{H_\tau}{W^2} \;<\; 0\,,
\]
so investors just above the threshold are most affected, and the effect
vanishes for $W \to \infty$.  Third, the distortion is
\emph{progressive}: it is largest for moderate-wealth investors (who
benefit most from the exemption relative to their wealth) and negligible
for the very wealthy.

\begin{remark}[Direction of distortion]
The threshold distortion runs in the \emph{opposite} direction from the
HARA distortion of Section~\ref{sec:beyond_crra}.  There, subsistence
consumption creates a floor wealth $H$ that \emph{reduces} risk-taking;
here, the tax exemption creates a shield $H_\tau$ that
\emph{increases} it.  Whether the net effect is conservative or
aggressive depends on the relative magnitude of $H$ and $H_\tau$.
\end{remark}

\subsection{Norwegian calibration (2026)}\label{sec:prog_norway}

The Norwegian wealth tax for 2026 has two brackets above a
threshold~\citep{Skatteetaten2026}:

\begin{center}
\begin{tabular}{lcccc}
\toprule
Bracket & Threshold & Municipal & State & Total rate \\
\midrule
Below threshold & $\bar{W}_1 = 1{,}900{,}000$ &
  \multicolumn{3}{c}{0\%} \\
Bracket 1 & $\bar{W}_1$ to $\bar{W}_2 = 21{,}500{,}000$ &
  0.35\% & 0.65\% & $\tau_1 = 1.00$\% \\
Bracket 2 & Above $\bar{W}_2$ &
  0.35\% & 0.75\% & $\tau_2 = 1.10$\% \\
\bottomrule
\end{tabular}
\end{center}
Thresholds are per person; married couples receive double.

\medskip
\noindent\textbf{Bracket~1 investors}
($1.9\text{M} < W \le 21.5\text{M}$ NOK).  The cumulative rebate is
$R_1 = \tau_1 \cdot \bar{W}_1 = 0.01 \times 1{,}900{,}000 =
19{,}000$ NOK per year.  With $r_f = 3\%$:
\[
  H_\tau = \frac{19{,}000}{0.03 - 0.01} = 950{,}000 \text{ NOK}\,.
\]
For an investor with $W = 5$M NOK, $H_\tau/W = 19\%$: the progressive
structure increases the optimal risky share by 19\% relative to the
proportional benchmark.  At $W = 15$M NOK, $H_\tau/W = 6.3\%$.

\medskip
\noindent\textbf{Bracket~2 investors} ($W > 21.5$M NOK).  The
cumulative rebate is $R_2 = R_1 + (\tau_2 - \tau_1)\,\bar{W}_2 =
19{,}000 + 0.001 \times 21{,}500{,}000 = 40{,}500$ NOK per year, and
\[
  H_\tau = \frac{40{,}500}{0.03 - 0.011} = 2{,}131{,}579 \text{ NOK}
  \approx 2.1\text{M NOK}\,.
\]
For $W = 30$M NOK, $H_\tau/W = 7.1\%$; for $W = 100$M, $H_\tau/W = 2.1\%$;
for $W = 1$B, $H_\tau/W = 0.2\%$.  The progressive distortion is thus
economically significant for moderate wealth but vanishes rapidly at the
top of the distribution.

\subsection{Interaction with non-homothetic preferences}
\label{sec:prog_hara}

Combining the progressive tax (\Cref{prop:progressive}) with HARA
preferences (\Cref{prop:hara_distortion} in Section~\ref{sec:beyond_crra}) gives a
portfolio share that reflects both the subsistence floor and the tax
shield:
\begin{equation}\label{eq:prog_hara}
  w^* \;=\; w^*_{\mathrm{neutral}}\;\frac{W - H + H_\tau}{W}\,,
\end{equation}
where $H = \zeta/(r_f - \tau_j)$ is the HARA floor wealth evaluated at
the after-tax risk-free rate of the investor's bracket, and $H_\tau$ is
the tax shield from~\eqref{eq:tax_shield}.

The two distortions partially offset: subsistence needs push the
portfolio toward safety ($-H$), while the tax exemption pushes it toward
risk ($+H_\tau$).  For Norwegian parameters with moderate subsistence
($\zeta = 100{,}000$ NOK per year, $r_f = 3\%$, $\tau_1 = 1\%$):
\[
  H = \frac{100{,}000}{0.02} = 5{,}000{,}000 \text{ NOK}\,,\qquad
  H_\tau = 950{,}000 \text{ NOK}\,.
\]
Since $H \gg H_\tau$, the HARA effect dominates for investors with
significant subsistence consumption.  However, for wealthy investors
with low subsistence-to-wealth ratios, the two effects are of comparable
magnitude and approximately cancel, restoring near-neutrality.

Figure~\ref{fig:shield_vs_hara} illustrates the opposing distortions.
The progressive tax shield pushes the risky share \emph{above} the CRRA
benchmark (blue curve), while the HARA subsistence floor pushes it
\emph{below} (red curve).  The combined effect (purple curve) lies
between: the shield partially offsets the HARA distortion, but the net
effect remains conservative because $H \gg H_\tau$ at Norwegian
parameters.  The shaded region between the HARA-only and combined curves
shows the magnitude of the progressive offset---the extent to which the
tax exemption ``recovers'' risk-taking capacity that the subsistence
floor removes.

\begin{figure}[t]
\centering
\begin{tikzpicture}[scale=1.0]
  \begin{axis}[
    width=12cm, height=8cm,
    xlabel={Wealth $W$ (NOK millions)},
    ylabel={Optimal risky share $w^*$},
    xmin=4, xmax=55,
    ymin=0, ymax=0.50,
    xtick={5,10,15,20,25,30,35,40,45,50},
    ytick={0,0.10,0.20,0.30,0.40,0.50},
    yticklabels={0,0.10,0.20,0.30,0.40,0.50},
    grid=major,
    grid style={gray!25},
    legend pos=south east,
    legend style={font=\small, draw=none, fill=white, fill opacity=0.8},
    every axis label/.style={font=\small},
    tick label style={font=\footnotesize},
    clip=false,
  ]

    \addplot[fill=green!12, draw=none, domain=5.3:55, samples=200,
      forget plot]
      {0.417*(1 - 4.05/x)} \closedcycle;
    \addplot[fill=white, draw=none, domain=5.3:55, samples=200,
      forget plot]
      {0.417*(1 - 5/x)} \closedcycle;

    \addplot[black, dashed, thick, domain=4:55, samples=2]
      {0.417};
    \addlegendentry{CRRA neutral}

    \addplot[blue!80!black, solid, very thick, domain=4:55, samples=200]
      {0.417*(1 + 0.95/x)};
    \addlegendentry{Progressive only ($+H_\tau$)}

    \addplot[red!80!black, solid, very thick, domain=5.3:55, samples=200]
      {0.417*(1 - 5/x)};
    \addlegendentry{HARA only ($-H$)}

    \addplot[violet!90!black, solid, very thick,
      domain=4.3:55, samples=200]
      {0.417*(1 - 4.05/x)};
    \addlegendentry{Combined ($-H + H_\tau$)}

    \draw[->, blue!80!black, thick]
      (axis cs:8,{0.417}) -- (axis cs:8,{0.417*(1+0.95/8)});
    \node[right, font=\footnotesize, text=blue!80!black]
      at (axis cs:8.2,{(0.417*(1+0.95/8)+0.417)/2+0.015})
      {shield $\uparrow$};

    \draw[->, red!80!black, thick]
      (axis cs:15,{0.417}) -- (axis cs:15,{0.417*(1-5/15)});
    \node[right, font=\footnotesize, text=red!80!black]
      at (axis cs:15.2,{(0.417+0.417*(1-5/15))/2})
      {floor $\downarrow$};

    \node[right, font=\footnotesize, text=green!50!black]
      at (axis cs:28,{0.417*(1-4.5/28)})
      {offset};

  \end{axis}
\end{tikzpicture}
\caption{Opposing portfolio distortions under progressive taxation with
HARA preferences.  The dashed line is the CRRA benchmark.  The blue
curve shows the progressive-only effect (CRRA investor with threshold):
the tax shield $H_\tau$ increases risk-taking.  The red curve shows the
HARA-only effect (proportional tax): subsistence floor $H$ reduces
risk-taking.  The purple curve combines both:
$w^* = w^*_{\mathrm{neutral}} \cdot (W - H + H_\tau)/W$.  The shaded
region shows the magnitude of the progressive offset.  Norwegian bracket~1
calibration: $H = 5$M~NOK, $H_\tau = 950{,}000$~NOK, same parameters as
Figure~\ref{fig:hara_distortion}.}
\label{fig:shield_vs_hara}
\end{figure}
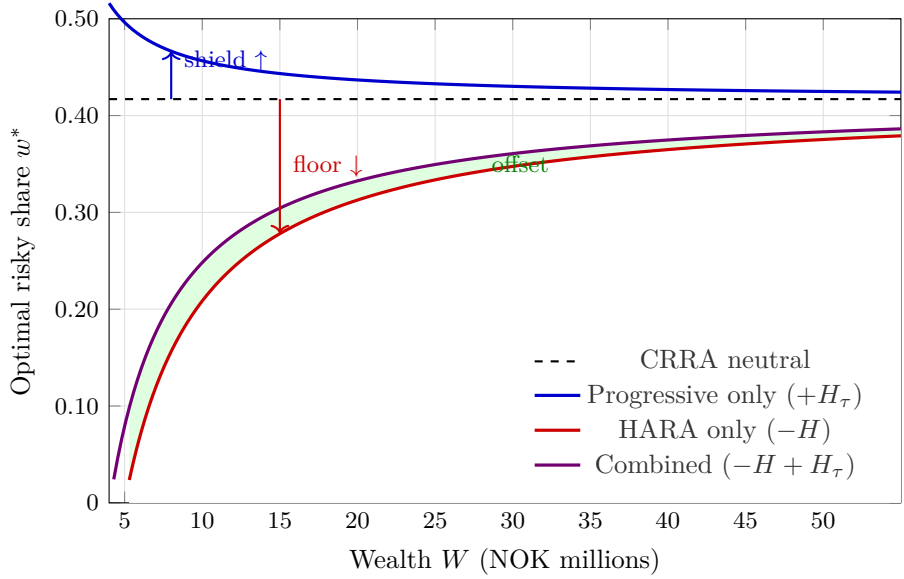

\subsection{Threshold bunching and behavioural responses}
\label{sec:prog_bunching}

The threshold at $\bar{W}_1$ creates a kink in the budget set: the
marginal tax rate jumps from 0 to $\tau_1$.  Standard bunching theory
\citep{Kleven2016} predicts that some investors will reduce their
taxable wealth to just below the threshold.  The empirical evidence is
substantial.

\citet{JakobsenEtAl2020} study the Danish progressive wealth tax
(rates $\sim$0.7--1.2\%, thresholds near the 98th percentile) using
administrative data.  They find sizable long-run elasticities of
taxable wealth at the top, driven by both real savings adjustments and
avoidance.  Their lifecycle model with ``residual wealth utility''
calibrates an elasticity of taxable wealth with respect to the
net-of-tax rate that implies meaningful revenue consequences from
threshold design.

\citet{GarbintiEtAl2024} provide a striking finding from the French
ISF: \emph{no bunching at pure tax-rate kinks}, but \emph{large, sharp
bunching at information-requirement thresholds}.  This suggests that the
portfolio distortion from bracket boundaries may be smaller than the
evasion response triggered by exemption thresholds---a distinction with
direct implications for Norwegian threshold design.

\citet{LondonoAvila2024} document clear bracket bunching in the
Colombian wealth tax, with two-fifths of the wealthiest 0.01\% hiding a
third of their wealth offshore.  The much larger behavioural response in
Colombia compared to Scandinavia underscores the role of third-party
reporting and enforcement in determining the practical relevance of
threshold effects.

In Norway, the combination of comprehensive third-party reporting (for
listed securities, bank deposits, and registered real estate) with
self-reporting for unlisted businesses creates a dual regime.  The
valuation discounts analysed in Section~\ref{sec:nonuniform} interact
with the threshold: an investor can reduce taxable wealth below
$\bar{W}_1$ by holding assets with low assessment fractions $\alpha_i$,
combining the non-uniform and progressive channels of non-neutrality.

\subsection{Policy implications}\label{sec:prog_policy}

The analysis yields three insights for threshold design.

First, a higher threshold reduces portfolio distortions for
moderate-wealth investors---precisely those most affected by the
progressive structure.  Raising $\bar{W}_1$ from 1.9M to, say, 5M NOK
would eliminate the tax shield for investors below 5M and reduce
$H_\tau/W$ for all remaining taxpayers.  This comes at a revenue cost
that depends on the wealth distribution; the concentrated ownership
structure of Norwegian wealth suggests the revenue loss may be moderate.

Second, the progressive structure eases the liquidity constraint.
Investors just above the threshold face a low average rate
$\bar{\tau}(W) = \tau_1(1 - \bar{W}_1/W)$, which mitigates the
consumption-saving and liquidity effects that would arise under a fully
proportional tax at the same marginal rate.  This is particularly
relevant for investors whose wealth is concentrated in illiquid assets
(real estate, private firms) where the wealth tax may force asset sales
or borrowing.

Third, the interaction between progressivity and the non-uniform
assessment of Section~\ref{sec:nonuniform} creates compound distortions.
An investor near the threshold faces strong incentives to hold assets
with low assessment fractions: primary housing ($\alpha = 0.25$) can
reduce taxable wealth below $\bar{W}_1$ for investors with gross wealth
well above the threshold.  To the extent that this encourages
owner-occupied housing over productive investment, the progressive
structure amplifies the misallocation from non-uniform assessment.

\section{Endogenous Labour Supply and Entrepreneurial Effort}
\label{sec:labour}

The preceding sections treat the investor's income as exogenous.  In
practice, the wealth tax may alter the incentive to supply labour or
exert entrepreneurial effort, creating a feedback loop between taxation,
wealth accumulation, and the labour-leisure margin.  This section
analyses the labour supply channel and its interaction with the
progressive structure of Section~\ref{sec:progressive}.

\subsection{Wealth effects under proportional taxation}
\label{sec:labour_proportional}

Consider an investor who chooses portfolio weight $w$, labour supply
$\ell$, and consumption $c$.  Labour earns income $y(\ell) = \omega\,\ell$
at a disutility cost $v(\ell)$, where $\omega$ is the wage rate and $v$
is increasing and convex.  The wealth dynamics become
\begin{equation}\label{eq:labour_dynamics}
  dW = \bigl[W\bigl(w(\mu - r_f) + r_f - \tau_w\bigr)
       + \omega\,\ell - c\bigr]\,dt + w\,W\,\sigma\,dZ\,.
\end{equation}
Under CRRA preferences over consumption and separable disutility of
labour, $U = u(c) - v(\ell)$ with $u(c) = c^{1-\gamma}/(1-\gamma)$,
the first-order conditions give
\begin{equation}\label{eq:labour_foc}
  v'(\ell) = \omega\,u'(c)\,.
\end{equation}
Since the wealth tax reduces lifetime wealth and hence consumption, it
raises marginal utility $u'(c)$, increasing optimal labour supply.

\begin{proposition}[Separability under proportional taxation]
\label{prop:labour_separability}
Under CRRA preferences with separable labour disutility and a
proportional wealth tax at rate $\tw$:
\begin{enumerate}
\item[(i)] The optimal portfolio weight is $w^* = (\mu -
  r_f)/(\gamma\sigma^2)$, independent of labour supply $\ell$.
\item[(ii)] The optimal labour supply satisfies $v'(\ell) = \omega\,
  u'(c^*)$, where $c^*$ is optimal consumption.  The wealth tax affects
  $\ell$ only through the income effect: it reduces $c^*$, raises
  $u'(c^*)$, and increases $\ell$.
\item[(iii)] The portfolio and labour decisions are decoupled.
\end{enumerate}
\end{proposition}

\begin{proof}
The proportional tax reduces all asset returns uniformly, leaving the
excess return $\mu - r_f$ and the variance $\sigma^2$ unchanged.  The
portfolio first-order condition depends only on these, not on $\ell$.
The labour first-order condition~\eqref{eq:labour_foc} depends on the
marginal utility of consumption, which is affected by the tax through
wealth but not through the portfolio weight.
\end{proof}

\begin{remark}[Analogy to a lump-sum tax]
For portfolio purposes, the proportional wealth tax is equivalent to a
reduction in the risk-free rate.  For labour supply purposes, it acts
like a lump-sum tax proportional to wealth: it reduces disposable
resources without distorting the marginal return to effort.
\end{remark}

\subsection{The threshold notch and labour supply}
\label{sec:labour_threshold}

The progressive structure of Section~\ref{sec:progressive} breaks the
separability.  At the threshold $\bar{W}_1$, the marginal tax rate jumps
from 0 to $\tau_1$, creating a \emph{notch} in the budget set.  An
investor whose wealth would place them just above $\bar{W}_1$ faces a
discrete choice: accumulate marginally more wealth and pay $\tau_1$ on
the excess, or reduce accumulation (by working less, consuming more, or
shifting into exempt assets) to remain below the threshold.

The notch induces a \emph{substitution effect} absent under
proportional taxation.  Near the threshold, the effective marginal tax
rate on wealth accumulation is locally infinite (a small increase in
$W$ from $\bar{W}_1 - \epsilon$ to $\bar{W}_1 + \epsilon$ triggers a
discrete tax liability on the entire excess).  This creates a dominated
region: investors with wealth slightly above $\bar{W}_1$ would be better
off at exactly $\bar{W}_1$.

The width of the dominated region depends on the tax rate and discount
rate.  In the Norwegian system, the annual tax liability at $W =
\bar{W}_1 + \Delta W$ is $\tau_1 \cdot \Delta W$.  For $\Delta W =
100{,}000$ NOK, this is 1,000 NOK per year---a modest amount that
limits bunching incentives for most investors.  However, the present
value of the perpetual tax stream $\tau_1 \cdot \Delta W / (r_f -
\tau_1)$ can be substantial: at $r_f = 3\%$, it equals $\Delta W / 2$,
meaning that every 100,000 NOK above the threshold costs 50,000 NOK
in present-value tax.

For entrepreneurs whose business wealth places them near the threshold,
the labour supply response interacts with the portfolio response:
reducing effort reduces both labour income and the value of the
business, potentially bringing total wealth below $\bar{W}_1$.  This
creates a compound distortion that does not arise under either
proportional taxation or exogenous income.

\subsection{Entrepreneurial effort and human capital}
\label{sec:labour_hk}

A distinctive feature of wealth taxation---relative to income or capital
gains taxation---is its treatment of returns to entrepreneurial effort.
\citet{GuvenEtAl2023} formalise this through heterogeneous returns on
capital: high-ability entrepreneurs earn higher returns, and a wealth
tax---unlike a capital income tax---does not penalise the higher return
directly.  The tax liability depends on the stock of wealth, not on
the flow of returns, creating a ``use it or lose it'' incentive: the
tax erodes idle or low-return capital while leaving high-return
capital relatively less burdened in after-tax terms.

In their calibrated model, replacing capital income taxation with
revenue-neutral wealth taxation raises average welfare by approximately
7\% of consumption-equivalent, driven by reallocation of capital from
low-productivity to high-productivity entrepreneurs.  A subsequent
paper by the same authors extends the framework to endogenous
innovation effort: when entrepreneurial productivity is a choice
variable, the wealth tax further incentivises effort because
entrepreneurs keep more of the upside compared to capital income
taxation.

This mechanism interacts with the Norwegian tax system's treatment of
intangible assets.  \citet{BjornebyEtAl2023} document a positive causal
relationship between wealth tax liability and employment in closely held
firms, using Norwegian register data and tax reforms over 2007--2017.
A 100,000 NOK increase in wealth tax liability is associated with
approximately 50,000 NOK in additional wage costs.  The mechanism
operates through two channels: an income effect (owners work harder to
compensate for the tax) and a portfolio reallocation effect (intangible
assets in non-traded firms are tax-exempt, incentivising owners to
invest in human capital within their businesses rather than in taxed
financial assets).

This finding is consistent with the non-uniform taxation framework of
Section~\ref{sec:nonuniform}: intangible business assets carry an
implicit assessment fraction of $\alpha = 0$, making them the most
tax-favoured asset class.  The wealth tax thus creates an incentive to
substitute financial capital for human capital---a reallocation that
may be efficiency-enhancing if entrepreneurs have comparative advantage
in their own firms.

\begin{remark}[Feedback to inelastic markets]
The labour supply response also feeds back into the general equilibrium
price channel of Section~\ref{sec:inelastic}.  If the threshold notch
(Section~\ref{sec:labour_threshold}) reduces entrepreneurial effort and
wealth accumulation, it moderates the aggregate flow $\Delta F$ and
hence the price impact.  Conversely, the income effect that increases
effort raises wealth accumulation and amplifies the flow.  The net
direction depends on which effect dominates---an empirical question
that the Norwegian evidence (modest positive income elasticity) suggests
resolves in favour of a small amplification.
\end{remark}

\subsection{Empirical evidence}\label{sec:labour_evidence}

The empirical literature on wealth effects and labour supply provides
context for calibrating the magnitude of these channels.

\citet{ScheuerSlemrod2021} survey cross-country evidence and report that
in Norway, households increase taxable labour income by approximately
0.01 NOK per additional NOK of wealth tax---a positive but small
elasticity consistent with a modest income effect.  In Switzerland and
Sweden, no significant earnings response to wealth taxation was found.

The lottery and inheritance literatures provide estimates of pure wealth
effects.  \citet{PicchioEtAl2018} study Dutch lottery winners and find
small but statistically significant reductions in labour earnings, with
winners reducing hours but rarely exiting the labour force.
Inheritance studies find that women reduce labour supply by
approximately 1.5 hours per week upon receiving a bequest, while men's
supply is largely unaffected.  These wealth effects run in the opposite
direction from the wealth tax (which \emph{reduces} wealth and hence
\emph{increases} labour supply), but their modest magnitude suggests
that the labour supply channel is unlikely to dominate portfolio and
price effects.

\citet{StraubWerning2020} overturn the classical Chamley--Judd result
that long-run capital taxation should be zero.  They show that when the
intertemporal elasticity of substitution is below one---the empirically
relevant range---the optimal long-run capital tax rate is positive and
significant.  This provides theoretical support for wealth taxation as
part of an optimal tax system, even when labour supply is endogenous.

\subsection{Policy implications}\label{sec:labour_policy}

The labour supply analysis reinforces and qualifies the findings of
earlier sections.

First, a proportional wealth tax distorts labour supply only through
the income effect, which operates in the ``benign'' direction of
increasing effort.  This is a welfare-relevant consideration that
partially offsets the deadweight loss from any portfolio distortion.

Second, the progressive structure introduces an additional substitution
effect at the threshold that can reduce effort for investors near
$\bar{W}_1$.  A higher threshold, as discussed in
Section~\ref{sec:prog_policy}, would eliminate this notch for
moderate-wealth investors and confine the substitution effect to
wealthier taxpayers for whom the threshold is less binding.

Third, the Norwegian evidence that wealth taxation increases
employment in closely held firms suggests that the human capital
reallocation channel may quantitatively dominate the standard labour
supply distortion---at least in a system with substantial asset-class
exemptions.  This is a second-best argument: the non-uniform assessment
that distorts portfolio choice (Section~\ref{sec:nonuniform}) may
simultaneously improve labour allocation by making human capital
investment tax-favoured.

Fourth, the ``use it or lose it'' mechanism of \citet{GuvenEtAl2023}
implies that the efficiency cost of wealth taxation is lower than that
of capital income taxation when entrepreneurial returns are
heterogeneous.  This is relevant for the design of the overall tax mix:
the marginal efficiency cost of raising revenue through the wealth tax
may be lower than commonly assumed, particularly at moderate rates.

\section{Application: A Global Minimum Wealth Tax}\label{sec:zucman}

Two recent proposals illustrate contrasting approaches to wealth
taxation at the top of the distribution.  \citet{Zucman2024}, developed
for the G20 Brazilian presidency, advocates a global minimum wealth tax
of 2\% on net wealth exceeding approximately \$1 billion, targeting the
roughly 3,000 individuals who currently pay effective tax rates of
0.2--0.3\% of their wealth; the estimated revenue is \$200--250 billion
annually.  In France, the National Assembly adopted a domestic variant
in February 2025: a minimum effective tax of 2\% on net wealth exceeding
\texteuro100 million, affecting approximately 1,800 households
\citep{FranceNA2025}.  Although both share the 2\% rate and Zucman's
intellectual framework, they differ in threshold, scope, and mechanism
in ways that activate different non-neutrality channels.  This section
evaluates both proposals through the lens of the extensions developed in
Sections~\ref{sec:beyond_crra}--\ref{sec:labour}.

\subsection{The global proposal: assessment through the extension framework}
\label{sec:zucman_assessment}

\paragraph{Non-homothetic preferences (Section~\ref{sec:beyond_crra}).}
For billionaires, wealth vastly exceeds any plausible subsistence floor:
$W \gg H$, so $H/W \approx 0$.  The HARA distortion of
\Cref{prop:hara_distortion} is negligible; CRRA is an excellent
approximation.  The portfolio distortion from non-homothetic preferences
is a non-issue for the target population.

\paragraph{Non-uniform assessment (Section~\ref{sec:nonuniform}).}
A defining design feature of the Saez--Zucman proposal is comprehensive
market-value assessment: all assets are taxed at market value with no
valuation discounts.  This eliminates the non-uniform channel entirely.
The portfolio distortion $\Delta w^*$ from
\Cref{prop:nonuniform} is zero when $\alpha_i = 1$ for all $i$.  By
contrast, the Norwegian system's extensive valuation discounts
(Section~\ref{sec:norway_calibration}) create substantial tilts toward
housing and away from deposits.  The Saez--Zucman design thus avoids a
major source of non-neutrality present in existing wealth taxes.

\paragraph{Progressive structure (Section~\ref{sec:progressive}).}
The extremely high threshold ($\sim$\$1 billion) means the tax shield
from \Cref{prop:progressive} is negligible relative to wealth.  At a
2\% rate with $r_f = 3\%$:
\[
  H_\tau = \frac{0.02 \times 10^9}{0.03 - 0.02} = 2 \times 10^9\,,
\]
so $H_\tau/W = 2$ for $W = \$1$B (the threshold investor) but
$H_\tau/W = 0.02$ for $W = \$100$B.  While the threshold effect is
mechanically large at the boundary, the target population is
concentrated well above it.  Moreover, the proposal is designed as a
\emph{minimum} tax---a top-up on existing obligations---so the effective
threshold is even higher for investors already paying wealth or capital
taxes.

\paragraph{Inelastic markets (Section~\ref{sec:inelastic}).}
This is the channel most relevant to the proposal.  Billionaires hold
concentrated equity positions, often representing controlling stakes in
publicly traded firms.  As a conservative upper bound (assuming
the full liability is met by equity sales; see
Section~\ref{sec:imh_calibration} for the payment-hierarchy refinement),
the aggregate tax-induced flow is
\[
  \frac{\Delta F}{P_{\mathrm{eq}}}
  = \tau_w \cdot \phi_B\,,
\]
where $\phi_B$ is the fraction of global equity market capitalisation
held by the $\sim$3,000 affected individuals.  Global billionaire
wealth is approximately \$14 trillion (2024), of which roughly 60\%
is in equities.  Against a global equity market capitalisation of
$\sim$\$110 trillion, $\phi_B \approx 8.4/110 \approx 0.08$.  The
annual flow-to-capitalisation ratio is
\[
  \frac{\Delta F}{P_{\mathrm{eq}}}
  = 0.02 \times 0.08 = 0.0016 = 0.16\%\,.
\]
Under the inelastic markets hypothesis with $M = 5$, the annual price
impact is
\[
  \frac{\Delta P}{P} = -M \cdot 0.16\% = -0.8\%\,.
\]
This is a non-trivial but modest annual effect.  Compounded over a
holding period and borne by \emph{all} equity holders (not only the
taxed billionaires), it represents a general equilibrium cost of the
proposal that the partial-equilibrium revenue estimate does not capture.

Two caveats apply.  First, the estimate assumes that billionaires
liquidate equities to pay the tax, whereas in practice they may use
dividends, borrow against holdings, or sell other assets.  Second, the
multiplier $M$ is estimated for the US equity market; the effective
multiplier for the concentrated, often illiquid holdings typical of
billionaire portfolios may differ substantially.

\paragraph{Labour supply (Section~\ref{sec:labour}).}
The income effect of a 2\% annual tax on multi-billion-dollar wealth
is large in absolute terms (\$20M per year on \$1B) but small relative
to the wealth base.  The ``use it or lose it'' mechanism is directly
relevant: a 2\% annual tax erodes the wealth of passive heirs
($\sim$20\% per decade in real terms) while imposing a lower effective
burden on high-return entrepreneurs.  This selection effect is a
feature, not a bug, of the proposal: it shifts wealth from low-return
to high-return holders over time.

\subsection{The French variant: a national minimum wealth tax}
\label{sec:france_assessment}

The French proposal shares the 2\% rate and minimum-tax logic of the
global blueprint but differs in three dimensions that alter the
non-neutrality profile.

\paragraph{Lower threshold (\texteuro100M vs \$1B).}
The progressive channel is substantially more relevant.  At the French
threshold with $r_f = 3\%$:
\[
  H_\tau = \frac{0.02 \times 10^8}{0.03 - 0.02} = 2 \times 10^8
  = \text{\texteuro200M}\,,
\]
so $H_\tau/W = 2$ for the threshold investor ($W = \text{\texteuro100M}$)
and $H_\tau/W = 0.2$ even at $W = \text{\texteuro1B}$.
By \Cref{prop:progressive}, the tax shield increases risk-taking for
investors near the threshold, and the effect remains quantitatively
important across a larger fraction of the affected population than under
the global proposal.  Combined with the HARA channel---still small at
this wealth level ($H/W < 0.01$) but non-negligible in percentage
terms---the net distortion is larger than for billionaires.

\paragraph{Minimum tax mechanism.}
The French proposal is explicitly a top-up: investors already paying at
least 2\% of their wealth through income, capital gains, and property
taxes owe nothing additional.  This creates a complementarity between
wealth taxation and income realisation that is absent from a pure
wealth tax.  In our framework, the effective wealth tax rate is
\[
  \tau_w^{\text{eff}}(W, Y) = \max\!\bigl(0,\;
    0.02 - T_{\text{existing}}(Y)/W\bigr)\,,
\]
where $T_{\text{existing}}(Y)$ is total existing tax paid and $Y$
denotes realised income.  This has two implications for portfolio choice.
First, investors with high-dividend or high-turnover portfolios face
a lower effective wealth tax, which reduces the non-uniform assessment
channel for yield-generating assets.  Second, the top-up structure
creates an incentive to realise income up to the 2\% threshold,
reversing the usual lock-in effect of capital gains taxation---a
behavioural response outside the scope of our static framework but
worth noting.

\paragraph{Scope, portfolio geography, and market impact.}
The tax applies to the \emph{global} assets of French tax residents,
not only to French-domiciled holdings.  With $\sim$1,800 affected
households, the aggregate wealth subject to the tax is approximately
\texteuro400--500 billion, of which roughly 60\% is in equities
($\sim$\texteuro270 billion).  These equity holdings are
internationally diversified: French ultra-high-net-worth portfolios
typically have substantial exposure to US, European, and emerging
markets, with a home-bias share in French-listed equities that we
denote $h$.  The tax-induced selling flow is therefore split across
multiple markets.

Against global equity capitalisation of $\sim$\texteuro100 trillion,
the aggregate flow ratio is modest:
\[
  \frac{\Delta F_{\text{global}}}{P_{\text{eq}}^{\text{global}}}
  = \frac{0.02 \times 270}{100{,}000} \approx 0.005\%\,,
\]
an order of magnitude smaller than the global Saez--Zucman proposal.
However, with home bias the flow is disproportionately concentrated
in French equities.  If $h = 0.4$ (a plausible estimate for
concentrated family holdings), the Euronext-specific flow ratio is
\[
  \frac{\Delta F_{\text{FR}}}{P_{\text{eq}}^{\text{FR}}}
  = \frac{0.02 \times 0.4 \times 270}{3{,}500} \approx 0.06\%\,,
\]
where \texteuro3.5 trillion is the Euronext Paris capitalisation.
With a multiplier $M = 5$, the local price impact is
$\Delta P/P \approx -0.3\%$---smaller than the $-0.8\%$ global
estimate, reflecting the dilution of selling pressure across
international markets.  For individual stocks in which affected
taxpayers hold controlling stakes, the local impact could be
substantially larger, particularly if stakes are illiquid and
represent a high fraction of the free float.

The \citet{Brulhart2022} migration channel is also more salient for a
national tax: the 24\% inter-cantonal response they estimate for the
Swiss wealth tax is likely a lower bound for cross-border migration
within the EU's single market, absent coordinated enforcement.  The
combination of global asset coverage with national jurisdiction creates
a tension: the tax base is mobile even if the assets themselves are not.

\subsection{Comparison with existing wealth taxes}
\label{sec:zucman_comparison}

Table~\ref{tab:comparison} summarises how the extension channels apply
to the Norwegian system, the global Saez--Zucman proposal, and the
French variant.

\begin{table}[t]
\centering
\small
\caption{Non-neutrality channels under three wealth tax designs.}
\label{tab:comparison}
\begin{tabular}{lccc}
\toprule
Channel & Norwegian & Global (Zucman) & French variant \\
\midrule
Non-homothetic (HARA)
  & Moderate
  & Negligible
  & Small \\
  & ($H/W \sim 0.1$--$0.5$)
  & ($H/W \approx 0$)
  & ($H/W < 0.01$) \\[4pt]
Non-uniform assessment
  & Large
  & Zero
  & Reduced \\
  & (tilts to housing)
  & (market values)
  & (top-up offsets) \\[4pt]
Progressive threshold
  & Significant
  & Negligible
  & Significant \\
  & (near NOK~1.9M)
  & (threshold $\sim$\$1B)
  & (threshold \texteuro100M) \\[4pt]
Inelastic markets
  & Small, local
  & Moderate, global
  & Small--moderate \\
  & ($\phi \approx 0.2$)
  & ($\phi_B \approx 0.08$)
  & (diluted globally) \\[4pt]
Labour supply
  & Positive
  & ``Use it or lose it''
  & ``Use it or lose it'' \\
  & (human capital)
  & (selection effect)
  & + income realisation \\
\bottomrule
\end{tabular}
\end{table}

The comparison reveals a spectrum of design trade-offs.  The Norwegian
system imposes a low rate (1.0--1.1\%) on a broad population but
introduces substantial non-neutralities through valuation discounts,
progressive thresholds, and debt deductibility.  The global
Saez--Zucman proposal imposes a higher rate (2\%) on a very narrow
population and eliminates most non-neutrality channels through
comprehensive market-value assessment and an extremely high threshold;
the residual distortion is concentrated in the inelastic markets
channel.  The French variant occupies an intermediate position: its
lower threshold activates the progressive channel for a larger fraction
of affected taxpayers, while the top-up mechanism partially mitigates
the non-uniform assessment channel by crediting existing taxes against
the wealth tax liability.  All three designs share an irreducible
non-neutrality from the inelastic markets channel---an unavoidable
consequence of any tax that induces portfolio rebalancing in a market
with finite liquidity.

\subsection{Implementation challenges}\label{sec:zucman_implementation}

The practical feasibility of the proposal depends on two conditions
outside our model.

First, \emph{international coordination}.  Without universal adoption,
the proposal creates incentives for tax-base migration.
\citet{Brulhart2022} estimate that 24\% of the Swiss wealth tax
response comes from inter-cantonal migration; at the global level, the
relevant margin is jurisdiction shopping among sovereign states.
\citet{LondonoAvila2024} document that two-fifths of the wealthiest
Colombians hide a third of their wealth offshore, underscoring the
enforcement challenge.

Second, \emph{valuation of illiquid assets}.  The proposal requires
annual market-value assessment of private companies, art, real estate,
and other assets that do not trade on liquid markets.  This is precisely
the valuation problem that drives the non-uniform assessment of
Section~\ref{sec:nonuniform}: Norwegian valuation discounts exist in
part because market values are difficult to determine.  For listed
equities and bank deposits, the problem is solved by third-party
reporting; for private wealth, it remains the binding constraint on
implementation.

\section{Tax-Induced Migration}\label{sec:migration}

The preceding sections analyse how investors respond to wealth taxation
by adjusting portfolios, labour supply, and asset prices.  A more
radical response is to change tax jurisdiction entirely.  Under
European-style domicile-based taxation, a change of residence is
sufficient to eliminate the wealth tax liability; the progressive
threshold of Section~\ref{sec:progressive} then implies a
\emph{participation margin} at which investors exit the tax base
altogether.

\subsection{Migration as a threshold response}
\label{sec:migration_theory}

Consider an investor with wealth $W > \bar{W}_1$ (the first taxable
threshold) who faces a heterogeneous migration cost $c_i > 0$,
reflecting the disutility of relocation: social ties, language, business
location-specificity, family considerations, and regulatory frictions.
Under domicile-based taxation, the investor migrates if the present
value of future wealth tax liabilities exceeds the migration cost.

Within the progressive framework of Section~\ref{sec:progressive}, the
present value of the tax flow for an investor in bracket $j$ is
\begin{equation}\label{eq:migration_pv}
  \mathrm{PV}_{\tau}(W) = \frac{\tau_j \, W - R_j}{r_f - g + \tau_j}\,,
\end{equation}
where $g$ is the expected real growth rate of wealth and the denominator
reflects that the tax base grows at rate $g - \tau_j$ (net of wealth
erosion from the tax itself).  The investor migrates when
$\mathrm{PV}_{\tau}(W) > c_i$.

This yields a migration threshold
\begin{equation}\label{eq:migration_threshold}
  W_i^* = \frac{c_i(r_f - g + \tau_j) + R_j}{\tau_j}\,.
\end{equation}
Investors with $W > W_i^*$ migrate; those with $W < W_i^*$ stay.
Several features are immediate.

First, the migration threshold is \emph{decreasing} in $\tau_j$:
higher wealth tax rates lower the threshold at which migration becomes
attractive, expanding the set of investors who leave.  This is the
intensive margin that drives the empirical relationship between wealth
tax rates and out-migration.

Second, the progressive rebate $R_j$ \emph{raises} the migration
threshold, since part of the tax liability is offset by the cumulative
exemption.  A higher basic exemption ($\bar{W}_1$) therefore has a
dual effect: it eliminates the participation margin for investors near
the threshold (Section~\ref{sec:progressive}) \emph{and} raises the
migration cost--benefit ratio for investors above it.

Third, migration costs are high for most wealthy individuals.
\citet{Young2016} finds that US millionaire interstate migration is
just 2.4\% annually---below the general population rate of 2.9\%---even
though several US states levy no income tax.  The explanation is
\emph{embeddedness}: location-specific social capital, business ties,
family considerations, and lifecycle constraints make physical
relocation costly regardless of the tax saving.  These frictions apply
in any institutional context.

In the European setting, an additional asymmetry amplifies the
migration channel.  Under domicile-based taxation, a change of
residence is sufficient to eliminate the entire wealth tax liability;
the effective migration cost is $c_i$ alone.  Under a
citizenship-based system---as in the United States---the effective
cost becomes $c_i + \mathrm{PV}_{\text{exit}}$, where
$\mathrm{PV}_{\text{exit}}$ includes the present value of continued
tax obligations and any exit tax on unrealised gains.  This structural
difference suggests that European wealth tax migration responses may
exceed those observed in the US, even though the embeddedness frictions
identified by \citeauthor{Young2016} operate in both contexts.

Figure~\ref{fig:migration_threshold} illustrates the mechanism in
$(W, c_i)$ space.  Each line represents the indifference locus
$c_i = \mathrm{PV}_{\tau}(W)$ for a given tax rate: investors below the
line (low migration cost relative to tax burden) migrate; those above
stay.  All lines emanate from the exemption threshold
$(\bar{W}_1, 0)$, since at the threshold the tax liability is zero and
no migration incentive exists regardless of rate.  A tax increase from
0.85\% to 1.0\% rotates the line upward, expanding the ``migrate''
region: the shaded wedge between the two lines represents investors
who were previously in the ``stay'' region but are pushed into
migration by the rate increase.  At higher rates (2\%, the Zucman
proposal), the ``migrate'' region expands substantially, underscoring
the coordination argument of Section~\ref{sec:zucman}.

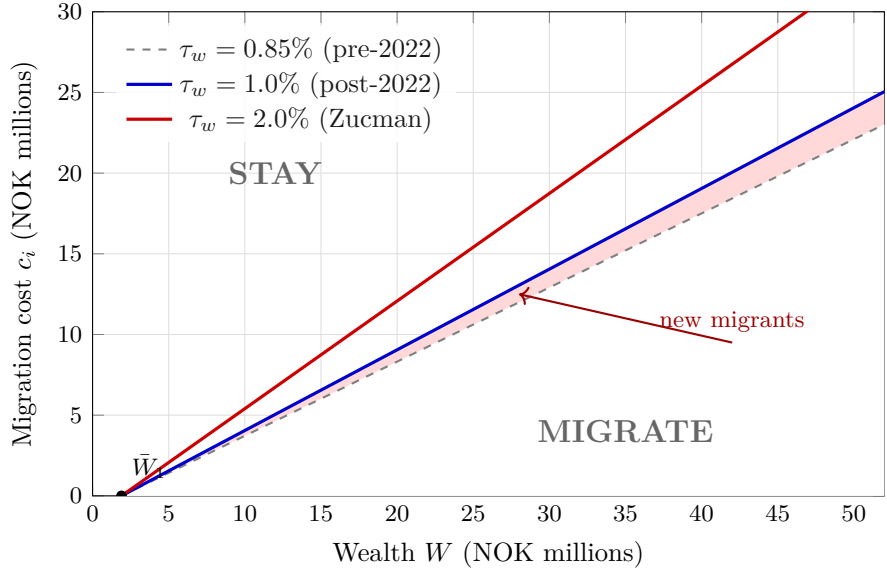
\begin{figure}[t]
\centering
\begin{tikzpicture}[scale=1.0]
  \begin{axis}[
    width=12cm, height=8cm,
    xlabel={Wealth $W$ (NOK millions)},
    ylabel={Migration cost $c_i$ (NOK millions)},
    xmin=0, xmax=52,
    ymin=0, ymax=30,
    xtick={0,5,10,15,20,25,30,35,40,45,50},
    ytick={0,5,10,15,20,25,30},
    grid=major,
    grid style={gray!25},
    legend pos=north west,
    legend style={font=\small, draw=none, fill=white, fill opacity=0.9},
    every axis label/.style={font=\small},
    tick label style={font=\footnotesize},
    clip=true,
  ]

    \addplot[fill=red!15, draw=none, domain=1.9:52, samples=100,
      forget plot]
      {0.5*x - 0.95} \closedcycle;
    \addplot[fill=white, draw=none, domain=1.9:52, samples=100,
      forget plot]
      {0.4595*x - 0.873} \closedcycle;

    \addplot[black!50, dashed, thick, domain=1.9:52, samples=2]
      {0.4595*x - 0.873};
    \addlegendentry{$\tau_w = 0.85\%$ (pre-2022)}

    \addplot[blue!80!black, solid, very thick, domain=1.9:52, samples=2]
      {0.5*x - 0.95};
    \addlegendentry{$\tau_w = 1.0\%$ (post-2022)}

    \addplot[red!80!black, solid, very thick, domain=1.9:52, samples=2]
      {0.667*x - 1.267};
    \addlegendentry{$\tau_w = 2.0\%$ (Zucman)}

    \node[font=\large\bfseries, text=black!60]
      at (axis cs:12,20) {STAY};
    \node[font=\large\bfseries, text=black!60]
      at (axis cs:35,4) {MIGRATE};

    \draw[->, thick, red!60!black]
      (axis cs:42,9.5)
      node[above, font=\footnotesize, text=red!60!black]
      {new migrants}
      -- (axis cs:28,12.5);

    \fill[black] (axis cs:1.9,0) circle (2pt);
    \node[above right, font=\footnotesize] at (axis cs:1.9,0.5)
      {$\bar{W}_1$};

  \end{axis}
\end{tikzpicture}
\caption{Migration decision in $(W, c_i)$ space.  Each line shows the
indifference locus $c_i = \mathrm{PV}_{\tau}(W)$: investors below the
line migrate; those above stay.  A tax increase from 0.85\% to 1.0\%
rotates the line upward, expanding the ``migrate'' region (shaded
wedge).  The 2\% Zucman rate further expands it.  All lines emanate
from the exemption threshold $\bar{W}_1 = 1.9$M~NOK.  Under
citizenship-based taxation (US), the effective cost shifts upward by
$\mathrm{PV}_{\text{exit}}$, compressing the ``migrate'' region.
Calibration: $r_f = 3\%$, $g = 2\%$.}
\label{fig:migration_threshold}
\end{figure}

\subsection{The Norwegian case: 2022--2024}
\label{sec:norway_migration}

Norway provides a contemporary case study that illustrates both the
migration response and the identification challenges.  The Støre
government increased the effective wealth tax burden through several
simultaneous channels beginning in 2022: the wealth tax rate rose
from 0.85\% to 1.1\%, the effective tax rate on dividends and capital
gains increased from approximately 31.7\% to 37.8\% (through the
\emph{oppjusteringsfaktor}, the multiplicative upward adjustment
applied before the 22\% ordinary income tax rate), and the five-year
expiration period for the exit tax on unrealised gains was abolished
in November 2022.

The behavioural response was swift.  According to the Norwegian Ministry
of Finance, 82 high-wealth individuals with combined net wealth of
approximately NOK~46 billion relocated in 2022--2023, more than the
previous 13 years combined.  The Civita think tank reports that 261
residents with wealth above NOK~10 million left in 2022, roughly double
the pre-reform baseline.  Switzerland---which has low cantonal wealth
taxes, no capital gains tax, and established Norwegian expatriate
communities---was the primary destination.

However, the identification of a \emph{wealth tax} effect is confounded
by at least three concurrent factors.

\paragraph{Confounding with dividend and capital gains taxation.}
The effective dividend tax rate increased by approximately 6 percentage
points simultaneously with the wealth tax increase.  For business
owners who extract dividends to cover wealth tax liabilities---a common
pattern documented by \citet{BjornebyEtAl2023}---the two taxes interact:
the wealth tax creates a liquidity need that triggers dividend
extraction, which is itself taxed at the higher rate.  Separating the
marginal contribution of each tax requires variation in one holding the
other constant, which the 2022 reform does not provide.

\paragraph{Confounding with the exit tax.}
The abolition of the five-year exit tax expiration in November 2022
changed the calculus for investors with large unrealised capital gains.
Under the old regime, emigration followed by a five-year holding period
eliminated both the wealth tax and the capital gains tax on accumulated
gains---a powerful joint incentive.  Some high-profile departures
occurred just before the November deadline, suggesting that the exit
tax change, rather than the wealth tax increase alone, was the binding
margin for investors with concentrated equity positions.

\paragraph{Confounding with political and structural factors.}
The change from a conservative to a centre-left government in 2021
brought broader shifts in economic policy, regulatory approach, and
political rhetoric.  Wealthy individuals may respond not only to
enacted tax changes but also to perceived policy direction: the
expectation of further increases, proposed regulations on private equity,
and public debate about wealth concentration may jointly affect the
relocation decision.  Disentangling these channels from the tax rate
itself requires careful empirical work beyond the scope of a
theoretical framework.

Figure~\ref{fig:norway_emigration} summarises the identification
challenge visually: the emigration spike in 2022 coincides with
simultaneous changes in the wealth tax rate, the effective dividend tax
rate, and the exit tax regime.  No single-variable causal attribution
is possible from the time series alone.

\begin{figure}[t]
\centering
\begin{tikzpicture}[scale=1.0]
  \begin{axis}[
    width=12cm, height=7.5cm,
    xlabel={Year},
    ylabel={Emigrants with wealth $>$10M NOK},
    axis y line*=left,
    ymin=0, ymax=320,
    xmin=2017.5, xmax=2024.5,
    xtick={2018,2019,2020,2021,2022,2023,2024},
    xticklabel style={/pgf/number format/1000 sep={}},
    ytick={0,50,100,150,200,250,300},
    grid=major,
    grid style={gray!20},
    every axis label/.style={font=\small},
    tick label style={font=\footnotesize},
    ylabel style={blue!70!black},
    clip=false,
  ]
    \addplot[ybar, bar width=20pt, fill=blue!30, draw=blue!60!black]
      coordinates {
      (2018,115) (2019,120) (2020,105) (2021,130)
      (2022,261) (2023,254) (2024,180)
    };

    \draw[red!70!black, very thick, dashed]
      (axis cs:2021.75,0) -- (axis cs:2021.75,310);
    \node[above, font=\tiny, text=red!70!black, rotate=90]
      at (axis cs:2021.6,155) {Støre govt};

    \node[font=\scriptsize, text=red!80!black, anchor=south east]
      at (axis cs:2021.55,275) {$\tau_w = 0.85\%$};
    \node[font=\scriptsize, text=red!80!black, anchor=south west]
      at (axis cs:2021.95,275) {$\tau_w = 1.1\%$};

  \end{axis}

  \begin{axis}[
    width=12cm, height=7.5cm,
    axis y line*=right,
    axis x line=none,
    ylabel={Eff.\ dividend tax rate (\%)},
    ymin=25, ymax=42,
    xmin=2017.5, xmax=2024.5,
    xtick=\empty,
    ytick={26,28,30,32,34,36,38,40},
    ylabel style={orange!70!black},
    every axis label/.style={font=\small},
    tick label style={font=\footnotesize},
    clip=false,
  ]
    \addplot[orange!80!black, very thick, mark=triangle*,
      mark size=2.5pt]
      coordinates {
      (2018,30.6) (2019,31.7) (2020,31.7) (2021,31.7)
      (2022,35.2) (2023,37.8) (2024,37.8)
    };
    \node[right, font=\scriptsize, text=orange!80!black]
      at (axis cs:2024,37.8) {37.8\%};

    \draw[->, thick, black!60] (axis cs:2023.2,30)
      -- (axis cs:2022.2,26.5);
    \node[right, font=\tiny, text=black!60] at (axis cs:2023.1,30.5)
      {5-yr exit rule abolished};

  \end{axis}

  \node[anchor=north west, draw=none, fill=white, fill opacity=0.9,
    text opacity=1, inner sep=4pt, font=\small]
    at (0.6,5.2) {%
      \begin{tabular}{@{}l@{}}
        \raisebox{1pt}{\tikz[baseline=-0.5ex]{%
          \fill[fill=blue!30, draw=blue!60!black]
          (0,0) rectangle (8pt,8pt);}}\;Emigrants ($>$10M NOK)\\[3pt]
        \raisebox{1pt}{\tikz[baseline=-0.5ex]{%
          \draw[orange!80!black, very thick] (0,3pt) -- (12pt,3pt);
          \fill[orange!80!black] (6pt,7.5pt) -- (3pt,1pt) -- (9pt,1pt) -- cycle;%
        }}\;Eff.\ dividend tax (\%, right axis)
      \end{tabular}%
    };
\end{tikzpicture}
\caption{Norwegian high-wealth emigration and concurrent tax changes,
2018--2024.  Bars show annual emigrants with wealth above NOK~10 million
(left axis; Civita estimates; 2024 figure is preliminary).  The line
shows the effective dividend tax rate (right axis).  The wealth tax rate
is annotated at the reform boundary ($0.85\% \to 1.1\%$).
The emigration spike in 2022 coincides with increases in both
rates and the abolition of the five-year exit tax expiration (arrow),
illustrating the identification challenge.}
\label{fig:norway_emigration}
\end{figure}
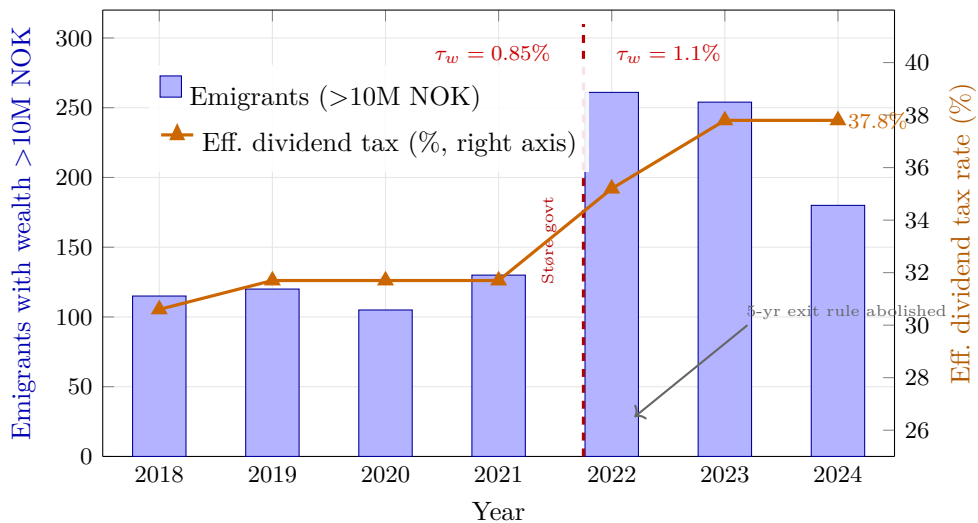

\subsection{Decomposing behavioural responses}
\label{sec:migration_decomposition}

The migration margin interacts with other behavioural channels.
\citet{Brulhart2022} decompose the Swiss wealth tax response into
three components: taxpayer migration (24\%), capitalisation into
housing prices (21\%), and evasion or avoidance (55\%).
\citet{Seim2017} finds that approximately one-third of the Swedish
response reflects underreporting.  In the Norwegian context,
\citet{IaconoSmedsvik2024} exploit the Bø municipality experiment
---in which a single municipality unilaterally reduced its municipal
wealth tax rate from 0.85\% to 0.35\% in 2021---and find that
taxpayer mobility accounts for 79\% of the increase in taxable wealth.

These decompositions suggest a hierarchy of responses. Evasion and
avoidance (including legal restructuring) are the first margin, as
they carry lower personal costs than physical relocation.  Migration
is the margin of last resort---activated when the tax burden is
sufficiently large relative to migration costs, consistent with the
threshold model of equation~\eqref{eq:migration_threshold}.  The
\citet{AgrawalEtAl2025} finding that Spanish mobility responds to
large discrete tax differentials (Madrid as a regional tax haven) but
not to small inter-regional variation supports this interpretation:
migration requires a sufficiently large $\mathrm{PV}_{\tau}(W) - c_i$
gap.

\subsection{Implications for the framework}
\label{sec:migration_implications}

The migration channel has two implications for the non-neutrality
analysis.

First, it introduces a \emph{selection effect} on the remaining tax
base.  If the highest-wealth investors migrate, the population
remaining in the tax jurisdiction is truncated from above.  This
changes the composition of the investor base in the inelastic markets
analysis (Section~\ref{sec:inelastic}): the investors who remain may
hold less concentrated, more diversified portfolios, potentially
reducing the aggregate demand multiplier.  Paradoxically, Norway's
wealth tax revenues \emph{increased} from NOK~27 billion (2022) to a
projected NOK~34 billion (2025) despite the exodus, because the
broader middle-wealth tax base remained intact and asset values
appreciated.

Second, the threat of migration creates an implicit constraint on
the wealth tax rate.  In a Tiebout framework, jurisdictional
competition for mobile tax bases limits the feasible tax rate to the
level at which the marginal migrant is indifferent.  This is precisely
the coordination argument for the Saez--Zucman global proposal
(\Cref{sec:zucman}): a minimum tax enforced across jurisdictions
eliminates the migration margin by removing the outside option.  The
failure of unilateral national proposals---such as the French variant
of Section~\ref{sec:france_assessment}, rejected in part due to
concerns about capital flight---illustrates the coordination constraint.

A full welfare analysis of the migration channel requires an empirical
estimate of the migration cost distribution $F(c)$, which determines
both the revenue loss from emigration and the deadweight loss from
distorted location decisions.  \citet{KlevenEtAl2020} emphasise that
migration elasticities are not structural parameters: they depend on
population characteristics, jurisdiction size, the degree of
international tax coordination, and the availability of evasion
alternatives.  This is fundamentally an empirical identification
problem that our theoretical framework can frame but not resolve.

\section{Conclusion}\label{sec:conclusion}

This paper has investigated the robustness of wealth tax neutrality along
two dimensions: the breadth of conditions under which neutrality holds,
and the channels through which it breaks.

On the first dimension, we have shown that portfolio neutrality extends
well beyond the geometric Brownian motion and location-scale settings of
\citet{Froeseth2026}.  Under CRRA preferences, neutrality is preserved
in stochastic volatility models---including the Heston model and general
Markov diffusions---where it applies to all intertemporal hedging demands,
not just the myopic component (\Cref{prop:sv_neutrality,prop:markov_neutrality}).
It is also preserved under Epstein--Zin recursive utility, where the
separation of risk aversion from the elasticity of intertemporal
substitution does not affect the homogeneity that drives the result
(\Cref{prop:ez}).  The operative condition remains CRRA: non-homothetic
preferences such as HARA break neutrality by making the effective risk
aversion wealth-dependent (\Cref{prop:hara_distortion}).

On the second dimension, we have identified four channels through which
implemented wealth taxes depart from neutrality, even when investors
have CRRA preferences.  Non-uniform assessment across asset classes
tilts portfolios toward assets with lower valuation fractions, a
distortion we calibrate to the Norwegian system
(\Cref{prop:nonuniform}).  General equilibrium effects in inelastic
markets amplify fundamental price adjustments through a demand
multiplier, creating a wedge between partial and general equilibrium
outcomes (\Cref{sec:inelastic}).  Progressive threshold structures
introduce a tax shield that increases risk-taking for investors near
the exemption threshold, an effect that operates in the opposite
direction from HARA and partially offsets it when both are present
(\Cref{prop:progressive}).  Endogenous labour supply is separable from
portfolio choice under proportional taxation but introduces additional
distortions at progressive thresholds (\Cref{prop:labour_separability}).
At the extreme, the progressive threshold generates a participation
margin: tax-induced migration (\Cref{sec:migration}).  Under European
domicile-based taxation, relocation eliminates the entire wealth tax
liability once the present value of that liability exceeds the
investor's migration cost.  The Norwegian experience after 2022
illustrates both the reality of this response and the difficulty of
isolating it from concurrent changes to dividend taxation, exit tax
rules, and the political environment.
The application to the Saez--Zucman global proposal and its French
national variant (\Cref{sec:zucman}) illustrates how design
choices---threshold level, scope, and coordination---activate
different subsets of these channels.

\subsection*{Directions for further work}

Several extensions remain open.

\paragraph{Theoretical.}
The multi-period consumption-saving problem with endogenous labour
supply has been studied in separate literatures but not in the
presence of a wealth tax.  A model that jointly determines
consumption, portfolio choice, and labour supply under progressive
wealth taxation would unify the channels analysed here.
General equilibrium with heterogeneous agents---in particular,
different wealth tax rates across brackets---would capture the
redistributive dynamics that a representative-agent framework cannot.
The interaction between wealth tax, capital gains tax, and income
tax in a joint optimal design problem is a natural next step, though
the presence of foreign investors subject to different tax regimes
adds substantial complexity.  Finally, extending the security market
fan of \citet{Froeseth2026} to multi-factor models would clarify how
the wealth tax interacts with factor premia.

\paragraph{Empirical.}
The non-uniform assessment channel can be tested directly using
Norwegian wealth registry data: the introduction and removal of
valuation discounts provide natural experiments for estimating
portfolio reallocation.  The progressive threshold channel predicts
bunching behaviour near the exemption boundary, which can be tested
using the administrative data employed by \citet{JakobsenEtAl2020}
for Denmark and \citet{GarbintiEtAl2024} for France.  The inelastic
markets channel yields testable predictions about asset price
responses to wealth tax reforms: \citet{Brulhart2022} provide a
template using Swiss cantonal variation, while the Norwegian system
offers variation across asset classes with different assessment
fractions.  On the labour supply side, the ``use it or lose it''
mechanism of \citet{GuvenEtAl2023} can be tested using the
relationship between entrepreneurial returns and wealth tax payments
in the Norwegian firm register.  The migration channel
(Section~\ref{sec:migration}) poses the hardest identification
challenge: \citet{JakobsenEtAl2024} estimate that a one percentage
point increase in the top wealth tax rate reduces the stock of wealthy
taxpayers by approximately 2\%, but as \citet{KlevenEtAl2020}
emphasise, migration elasticities are not structural parameters but
depend on population characteristics, jurisdiction size, the degree of
international tax coordination, and the availability of evasion
alternatives.  The Norwegian post-2022
experience, where wealth tax, dividend tax, and exit tax changes
occurred simultaneously alongside a change in political direction,
illustrates why isolating the wealth tax margin requires careful
empirical design---ideally exploiting variation in one instrument
while holding the others constant.  Finally, the empirical literature
on wealth taxation and firm behaviour is dominated by traditional
private firms; evidence on intangible-heavy, venture-capital-backed
firms---where book-value assessment creates the largest effective
discount and liquidity frictions are most acute---remains scarce,
despite these firms' prominence in the current policy debate.


\subsection*{Acknowledgements}
The author acknowledges the use of Claude (Anthropic) for assistance with
literature review, \LaTeX{} typesetting, mathematical exposition, and
editorial refinement, and Lemma (Axiomatic AI) for review and proof
checking. All substantive arguments, economic reasoning, and conclusions
are the author's own.

\appendix

\section{Leverage and debt deductibility: historical examples}
\label{app:leverage}

This appendix provides calibrated examples of the sheltering strategy
discussed in Section~\ref{sec:leverage}.

Under the system prevailing before Norway's proportional debt reduction
rules, debt was deductible at face value ($\beta_i = 1$) while assets
received valuation discounts.  Consider commercial property under the
old 45\% discount ($\alpha = 0.55$) financed at 80\% loan-to-value
($\ell = 0.80$):
\[
  \frac{\text{Taxable wealth}}{V} = 0.55 - 0.80 = -0.25.
\]
Each krone of leveraged commercial property \emph{reduced} the
investor's tax base by 0.25 kroner.  An investor holding NOK~10 million
in bank deposits (assessed at full value) could eliminate the associated
tax liability entirely by accumulating approximately NOK~40 million of
leveraged commercial property, whose negative taxable contribution
offsets the deposits.  Once the tax base reaches zero (its floor),
further leveraged property yields no additional tax benefit, creating a
kink in the effective marginal incentive.

For primary housing ($\alpha = 0.25$) with moderate leverage
($\ell = 0.60$), the sheltering capacity was even larger:
$0.25 - 0.60 = -0.35$ per krone of asset value.

Norway's proportional debt reduction rules now set $\beta_i = \alpha_i$
for discounted assets, eliminating negative tax bases.  For shares and
commercial property with a 20\% discount, the associated debt deduction
is reduced by 20\%, giving $\beta_i = 0.80$.  While this eliminates the
most egregious sheltering, the effective tax rate on equity in
discounted asset classes ($\tau_w \alpha_i$) remains below the statutory
rate, preserving a residual incentive to hold leveraged positions in
discounted assets.

\section{Time-scale dependence of stylised facts}
\label{app:timescale}

This appendix provides the detailed discussion of how the stylised facts
motivating the stochastic volatility extension vary across time horizons,
supplementing the summary in Section~\ref{sec:timescale}.

Among \citeauthor{Cont2001}'s (\citeyear{Cont2001}) eleven stylised
facts, the most dramatic departures from GBM are high-frequency
phenomena.  The inverse cubic power law for return tails
\citep{Gopikrishnan1999,Gabaix2003} attenuates at longer horizons via
aggregational Gaussianity, with a crossover to approximately Gaussian
tails at horizons of several weeks to a few months
\citep{BouchaudPotters2003}.  Volatility clustering weakens from
long-range dependence at daily frequency \citep{DingGrangerEngle1993}
to short-memory persistence at monthly and annual horizons, where the
Heston model's exponential autocorrelation is a reasonable approximation.

What \emph{does} persist at policy-relevant horizons is time-varying
expected returns \citep{Campbell1988}, regime-like volatility dynamics,
and the variance risk premium---precisely the features captured by the
general Markov diffusion framework of Section~\ref{sec:general_markov}.
The stochastic volatility extension is therefore valuable not because
high-frequency anomalies require it, but because it demonstrates that
neutrality extends to the empirically relevant features that \emph{do}
persist: time-varying investment opportunities, hedging demands, and
stochastic risk premia.  This also diminishes the urgency of extending
the result to jump-diffusion processes, since the power-law tails that
jumps would capture are primarily a high-frequency phenomenon.

\section{Empirical evidence on non-uniform assessment}
\label{app:empirical_nonuniform}

This appendix reviews the empirical literature on wealth tax effects
under non-uniform assessment, supplementing the summary in
Section~\ref{sec:nonuniform_empirical}.

\citet{Ring2024} provides the most direct empirical test of wealth tax
neutrality under non-uniform assessment.  Exploiting geographic
discontinuities in housing valuations created by Norway's 2010 reform,
Ring estimates the causal effect of the wealth tax on household saving
and portfolio allocation.  His central finding is that the wealth tax
has a positive effect on financial saving (elasticity of approximately
3.76), as households save to offset the tax, but the portfolio
\emph{composition} between risky and safe assets is unaffected when the
tax does not discriminate between them.  This is consistent with the
neutrality result under uniform assessment established in the main text.

However, when the assessment differentials create distinct after-tax
returns, portfolio distortions emerge.  \citet{FagerengGuisoRing2024}
study how investors respond to changes in the equity premium induced by
wealth tax reforms.  They find that households do adjust their risky
asset shares in response to differential tax treatment, but slowly:
full reoptimisation takes approximately five years.  This sluggish
adjustment is consistent with portfolio inertia and suggests that the
theoretical distortions derived in Section~\ref{sec:nonuniform}
represent a long-run equilibrium that is approached gradually.

\citet{Brulhart2022} study wealth tax responses in Switzerland, where
cantonal rate variation creates quasi-experimental conditions.  They
find that a one percentage point reduction in the wealth tax rate
increases reported taxable wealth by approximately 43\%, though much of
this response reflects taxpayer mobility between cantons and valuation
responses rather than real portfolio reallocation.

\bibliographystyle{plainnat}

\end{document}